\documentclass[runningheads]{llncs}
\usepackage{amsmath}%
\usepackage{amsfonts}%
\usepackage{amssymb}%
\usepackage{amsbsy}%
\usepackage{amscd}%
\usepackage{latexsym}%
\usepackage{float}%
\usepackage{array}%
\usepackage{ulem}%
\usepackage{epic}%
\usepackage{wrapfig}
\usepackage{url}

\DeclareMathAlphabet{\mathrm}    {OT1}{cmr}{m}{n}
\DeclareMathAlphabet{\mathrmbf}  {OT1}{cmr}{bx}{n}
\DeclareMathAlphabet{\mathrmit}  {OT1}{cmr}{m}{it}
\DeclareMathAlphabet{\mathrmbfit}{OT1}{cmr}{bx}{it}
\DeclareMathAlphabet{\mathsf}    {OT1}{cmss}{m}{n}
\DeclareMathAlphabet{\mathsfbf}  {OT1}{cmss}{bx}{n}
\DeclareMathAlphabet{\mathsfit}  {OT1}{cmss}{m}{sl}
\DeclareMathAlphabet{\mathttbf}  {OT1}{cmtt}{bx}{n}
\newcommand{\keywords}[1]{\par\addvspace\baselineskip\noindent\enspace\ignorespaces{\bfseries Keywords:\,}#1}
\begin{document}

\pagestyle{headings}
\title{Database Semantics}
\author{Robert E. Kent}
\institute{Ontologos}
\maketitle

\begin{abstract}
This paper,
the first step to connect relational databases with systems consequence (Kent~\cite{kent:iccs2009}), 
is concerned with the semantics of relational databases. 
It aims to to study system consequence in the logical/semantic system of relational databases. 
The paper, 
which was inspired by and which extends a recent set of papers on the theory of relational database systems 
(Spivak~\cite{spivak:sd}~\cite{spivak:fdm}), 
is linked with work on the Information Flow Framework (IFF~\cite{iff}) connected with the ontology standards effort (SUO), 
since relational databases naturally embed into first order logic. 
The database semantics discussed here is concerned with the conceptual level of database architecture.
We offer both an intuitive and technical discussion. 
Corresponding to the notions of primary and foreign keys, 
relational database semantics takes two forms: 
a distinguished form where entities are distinguished from relations, and 
a unified form where relations and entities coincide. 
The distinguished form corresponds to the theory presented in (Spivak~\cite{spivak:sd}). 
The unified form, 
a special case of the distinguished form, 
corresponds to the theory presented in (Spivak~\cite{spivak:fdm}). 
A later paper will discuss various formalisms of relational databases, 
such as relational algebra and first order logic,
and will complete the description of the relational database logical environment.
\end{abstract}
%


\section{Introduction}\label{introduction}

The author's ``Systems Consequence'' paper (Kent~\cite{kent:iccs2009}) 
is a very general theory and methodology for specification and inter-operation of systems of information resources. 
The generality comes from the fact that it is independent of the logical/semantic system (institution) being used. 
This is a wide-ranging theory, 
based upon ideas from 
information flow (Barwise and Seligman~\cite{barwise:seligman:97}), 
formal concept analysis (Wille and Ganter et al~\cite{ganter:wille:99}), 
the theory of institutions (Goguen et al~\cite{goguen:burstall:92}), and 
the lattice of theories notion (Sowa), 
for the integration of both formal and semantic systems independent of logical environment. 
In order to better understand the motivations of that paper and to be able more readily to apply its concepts, 
in the future it will be important to study system consequence in various particular logical/semantic systems. 
This paper aims to do just that for the logical/semantic system of relational databases. 
The paper, 
which was inspired by and which extends a recent set of papers on the theory of relational database systems 
(Spivak~\cite{spivak:sd},\cite{spivak:fdm}), 
is linked with work on the Information Flow Framework (IFF~\cite{iff}) connected with the ontology standards effort (SUO), 
since relational databases naturally embed into first order logic. 
We offer both an intuitive and technical discussion. 
Corresponding to the notions of primary and foreign keys, 
relational database semantics takes two forms: 
a distinguished form where entities are distinguished from relations, and 
a unified form where relations and entities coincide. 
The distinguished form corresponds to the theory presented in the paper (Spivak~\cite{spivak:sd}). 
We extend Spivak's treatment of tables 
from the static case of a single entity classification (type specification) 
to the dynamic case of classifications varying along infomorphisms. 
Our treatment of relational databases as diagrams of tables differs from 
Spivak's sheaf theory of databases. 
The unified form, 
a special case of the distinguished form, 
corresponds to the theory presented in the paper (Spivak~\cite{spivak:fdm})). 
The unified form has a graphical presentation, 
which corresponds to the sketch theory of databases (Johnson and Rosebrugh~\cite{johnson:rosebrugh:07}) 
and the resource description framework (RDF). 
This paper, 
which is the first step to connect relational databases with systems consequence, 
is concerned with the semantics of relational databases. 
A later paper will discuss various formalisms of relational databases, 
such as relational algebra and first order logic.
Section~\ref{relational:data:model} discusses the relational data model.
Section~\ref{tables} describes our representation for the table concept, 
both defining a category of tables,
and proving that this category is complete (joins exist) and cocomplete (unions exist).
Section~\ref{relational:databases} represents the relational database concept
as a diagram of tables linked by the generalization-specialization of projections.
Morphisms of relational databases are defined.
Canonical examples of both are discussed.
Finally,
section~\ref{summary:future:work} summarizes the results and gives some concluding remarks.

\section{Relational Data Model}\label{relational:data:model}

The paper defines an architectural semantics for the relational data model.
\footnote{Older architectures of data include the hierarchical model and network model.
Of these, nothing will be said.
A newer architecture of data, 
called the object-relation-object model, 
is a presentation form for the relational data model described here.}
All information in the relational model is represented within relations.
A relational database is a collection of relations (relational tables, or just tables). 
A table is represented as an array, organized into rows and columns. 
The rows are called the tuples (records) of the table, whereas
the columns are called the attributes of the table.
Both rows (tuples) and columns are unordered.
In the basic relational data model all the components can be resolved into sets and functions.
~\footnote{The basic relational data model is defined on the category $\mathrmbf{Set}$ of sets and functions.}

The basic relational building block is the data domain represented by an entity type $x \in X$,
where $X$ is the type set of an entity classification 
$\mathcal{E} = {\langle{X,Y,\models_{\mathcal{E}}}\rangle}$, 
whose instance set is a universe of data values $Y$ local to the database.
An entity instance $y \in Y$ is classified by an entity type $x \in X$ when $y \models_{\mathcal{E}} x$.
Within the classification $\mathcal{E}$ the entity type $x \in X$ represents its extent,
which is the domain of data values
$\mathrmbfit{ext}_{\mathcal{E}}(x) = \{ y \in Y \mid y \models_{\mathcal{E}} x \}$.
We extend the classification to generalized elements.
An indexed collection of entity types     $\{ (i,s_{i}) \mid i \in I, s_{i} \in X \}$ is called an $\mathcal{E}$-signature.
It is denoted by the pair ${\langle{I,s}\rangle}$ 
and represented as a map $I \xrightarrow{s} X$ from index set to entity type set.
An indexed collection of entity instances $\{ (j,t_{j}) \mid j \in J, t_{j} \in Y \}$ is called an $\mathcal{E}$-tuple.
A tuple represents an object; 
either a concrete, physical object or an abstract, conceptual object. 
It is denoted by the pair ${\langle{J,t}\rangle}$ 
and represented as a map $J \xrightarrow{t} Y$ from index set to the universe.
The indexing set is called the arity of the signature or tuple.
A $\mathcal{E}$-tuple ${\langle{J,t}\rangle}$ is classified by an $\mathcal{E}$-signature ${\langle{I,s}\rangle}$,
denote by $t \models_{\mathcal{E}} s$,
when they have the same arity $J = I$ 
and enjoy pointwise classification $t_{i} \models_{\mathcal{E}} s_{i}$ for all $i \in I$.
The extent of an $\mathcal{E}$-signature ${\langle{I,s}\rangle}$ is its tuple set 
$\mathrmbfit{tup}_{\mathcal{E}}(I,s) = \{ t \mid t \models_{\mathcal{E}} s \}$.

Let $\mathcal{T}$ be a relational table in a database based on the entity classification $\mathcal{E}$. 
An attribute of $\mathcal{T}$
is an ordered pair $(i,s_{i})$ 
consisting of an attribute name $i \in I$ and an entity type $s_{i} \in X$,
where $I$ is the arity of the table. 
The collection of attributes of $\mathcal{T}$ forms its schema ${\langle{I,s,X}\rangle}$, 
where ${\langle{I,s}\rangle}$ is an $\mathcal{E}$-signature.
A tuple of $\mathcal{T}$ is an $\mathcal{E}$-tuple
that is classified by the table signature ${\langle{I,s}\rangle}$.
Hence,
the tuple set of $\mathcal{T}$ is the set $\mathrmbfit{tup}_{\mathcal{E}}(I,s)$.
Each tuple of $\mathcal{T}$ must be uniquely identifiable by some combination (one or more) of its attribute values. 
This combination is referred to as the primary key. 
Without loss of generality, we assume that (primary) keys are single attributes.
In addition,
we conceptually separate the primary key attribute from the rest of the table and use it for indexing.
Hence, the table $\mathcal{T}$ is an indexed collection of $\mathcal{E}$-tuples 
$\mathcal{T} = \{ (k,\tau_{k}) \mid \tau_{k} \in \mathrmbfit{tup}_{\mathcal{E}}(I,s), k \in K \}$,
where $K$ is the set of primary keys of the table; 
that is, the table is represented as a map
$K \xrightarrow{\tau} \mathrmbfit{tup}_{\mathcal{E}}(I,s)$
from keys to tuples.

Here is an small example of a relational database for a company in unified form, 
which illustrates both primary keys ($\blacktriangle$) and foreign keys ($\vartriangle$).
It contains two relational tables,
an employee table {\bf{Emp}} and a department table {\bf{Dept}},
which are indexed by primary keys and linked by foreign keys.
\begin{center}
\vspace{-5pt}
\begin{tabular}{@{\hspace{-5pt}}c@{\hspace{12pt}}c}
\renewcommand{\tabcolsep}{6pt}
\tiny
\setlength{\extrarowheight}{2pt}
\begin{tabular}[t]{|l||l|l|l|} 								\hline
\textit{emp}:{\bf{Emp}} & \textit{name}:{\bf{Str}}	& \textit{addr}:{\bf{Str}} & \textit{dept}:$\mathbf{\dot{Dept}}$	\\ \hline
e1			& Plato			&	Greece	& $\dot{d1}$ 	\\ \hline
e2			& Aquinus		&	Italy		& $\dot{d2}$ 	\\ \hline
e3			& Decartes	&	France	& $\dot{d1}$ 	\\ \hline
\multicolumn{1}{l}{$\;\;\blacktriangle$} & \multicolumn{2}{c}{} & \multicolumn{1}{l}{$\;\;\vartriangle$} \\
\end{tabular}
&
\renewcommand{\tabcolsep}{6pt}
\tiny
\setlength{\extrarowheight}{2pt}
\begin{tabular}[t]{|l||l|l|}					\hline
\textit{dept}:{\bf{Dept}} & \textit{name}:{\bf{Str}} & \textit{mngr}:$\mathbf{\dot{Emp}}$ 
\\ \hline
d1			& Sales				& $\dot{e3}$ \\ \hline
d2			& Production	& $\dot{e2}$ \\ \hline
\multicolumn{1}{l}{$\;\;\blacktriangle$} & \multicolumn{1}{c}{} & \multicolumn{1}{l}{$\;\;\vartriangle$} 
\end{tabular}
\end{tabular}
\end{center}
In this example,
the entity (relation) types are {\bf{Dept}}, {\bf{Emp}} and {\bf{Str}}.
In the employee relational table {\bf{Emp}},
the arity is $\{ \textit{name}, \textit{addr}, \textit{dept} \}$,
the signature is $\{ (\textit{name},{\bf{Str}}), (\textit{addr},{\bf{Str}}), (\textit{dept},{\bf{Dept}}) \}$, and
the (primary) key set is $\{ e1, e2, e3 \}$.
Dotted items indicate relations (types or instances) being used as entities, since this is in unified form.

In the relational data model, there are three inherent integrity constraints: 
entity integrity, domain integrity, and referential integrity.
Entity integrity asserts that every table must have a primary key column in which each entry identifies its own row (tuple).
Domain integrity asserts that each data entry in a column must be of the type of that column. 
Entity and domain integrity are requirements for the distinguished form of database semantics.
Entity integrity says there must be a tuple function from the set of (primary) keys, and
domain integrity says that image tuples must be classified by the table signature ${\langle{I,s}\rangle}$.
Hence, entity and domain integrity assert the existence of the tuple or content function
$t : K \rightarrow \mathrmbfit{tup}_{\mathcal{E}}(I,s)$.
Referential integrity asserts that each entry in a foreign key column of a referencing table 
must occur in the primary key column of the referenced table.
Referential integrity is a requirement for the unified form of database semantics.
Referential integrity says there must be a function 
from a foreign key column of a referencing table 
to the primary key column of the referenced table.
Hence, 
referential integrity asserts the existence of the functions in the sketch interpretation of a relational database.

The information in a database is accessed by specifying queries, 
which use operations such as 
select to identify tuples, 
project to identify attributes, and 
join to combine tables. 
In this paper,
projection refers to a primitive generalization-specialization operation between pairs of relational tables
(they are specified by the database schema, project from joined table to components, or other),
whereas join is a composite operation on a linked collection of tables.
Selection is a special case of join,
which uses reference relations (tables).

\section{Tables}\label{tables}

A \emph{table} (database relation) 
$\mathcal{T} = {\langle{\mathcal{S},\mathcal{E},K,t}\rangle}$
has an underlying (simple) schema 
$\mathcal{S} = {\langle{X,I,s}\rangle}$ 
with
a set of entity types $X$ and an $X$-signature ${\langle{I,s}\rangle} \in (\mathrmbf{Set}{\downarrow}X)$,
an entity classification $\mathcal{E} = {\langle{X,Y,\models_{\mathcal{E}}}\rangle}$ 
with a common (entity) type set component $X \in \mathrmbf{Set}$ 
and a local universe of entity instances $Y \in \mathrmbf{Set}$,
a set $K$ of (primary) keys, and
a tuple function $t : K \rightarrow \mathrmbfit{tup}_{\mathcal{E}}(I,s)$
mapping keys to $\mathcal{E}$-tuples of type (signature) ${\langle{I,s}\rangle}$.
Equivalently,
it is an object in the the comma category of $\mathcal{E}$-tables 
$\mathcal{T} \in (\mathrmbf{Set}{\downarrow}\mathrmbfit{tup}_{\mathcal{E}})$.

A \emph{table morphism} (morphism of database relations)
${\langle{h,f,g,k}\rangle} :
\mathcal{T}_{1} = {\langle{\mathcal{S}_{1},\mathcal{E}_{1},K_{1},t_{1}}\rangle} \rightarrow 
{\langle{\mathcal{S}_{2},\mathcal{E}_{2},K_{2},t_{2}}\rangle} = \mathcal{T}_{2}$
consists of 
a (simple) schema morphism
${\langle{h,f}\rangle} : 
\mathcal{S}_{2} = {\langle{X_{2},I_{2},s_{2}}\rangle} \rightarrow 
{\langle{X_{1},I_{1},s_{1}}\rangle} = \mathcal{S}_{1}$
with
a function on entity types $f : X_{2} \rightarrow X_{1}$
and
an $X_{1}$-signature morphism
$h : {\scriptstyle\sum}_{f}(I_{2},s_{2}) = {\langle{I_{2},s_{2}{\cdot}f}\rangle} \rightarrow {\langle{I_{1},s_{1}}\rangle}$, 
an entity infomorphism 
${\langle{f,g}\rangle} : 
\mathcal{E}_{2} = {\langle{X_{2},Y_{2},\models_{\mathcal{E}_{2}}}\rangle} \rightleftarrows 
{\langle{X_{1},Y_{1},\models_{\mathcal{E}_{1}}}\rangle} = \mathcal{E}_{1}$
with a common (entity) type function component $f : X_{2} \rightarrow X_{1}$
and a universe (entity instance) function $g : X_{1} \rightarrow X_{2}$, and
a key function
$k : K_{1} \rightarrow K_{2}$,
which satisfy the condition
$k \cdot t_{2} = t_{1} \cdot \mathrmbfit{tup}(h,f,g)$,
where
$\mathrmbfit{tup}(h,f,g) 
\doteq \mathrmbfit{tup}_{\mathcal{E}_{1}}(h) \cdot \tau_{{\langle{f,g}\rangle}}(I_{2},s_{2})
= (h{\cdot}{(\mbox{-})})\cdot({(\mbox{-})}{\cdot}g) :
\mathrmbfit{tup}_{\mathcal{E}_{1}}(I_{1},s_{1}) \rightarrow \mathrmbfit{tup}_{\mathcal{E}_{2}}(I_{2},s_{2})$.~\footnote{Since 
the table tuple function 
embodies the entity/domain integrity constraints (Section~\ref{relational:data:model}),
this condition on morphisms asserts the preservation of data integrity.}
Table morphisms are illustrated in Figure~\ref{table:morphism}.
Here we see that table morphisms have the pleasing property that
corresponding entries in the source and target tables
satisfy the infomorphism condition from the theory of information flow
(Barwise and Seligman~\cite{barwise:seligman:97}).
%
%
\begin{figure}
\begin{center}
\begin{tabular}{@{\hspace{20pt}}c@{\hspace{45pt}}l}
\begin{tabular}{c}
\setlength{\unitlength}{0.5pt}
\begin{picture}(120,80)(0,0)
\put(0,80){\makebox(0,0){\footnotesize{$K_{2}$}}}
\put(120,80){\makebox(0,0){\footnotesize{$K_{1}$}}}
\put(-35,0){\makebox(0,0){\footnotesize{$\mathrmbfit{tup}_{\mathcal{E}_{2}}(I_{2},s_{2})$}}}
\put(150,0){\makebox(0,0){\footnotesize{$\mathrmbfit{tup}_{\mathcal{E}_{1}}(I_{1},s_{1})$}}}
\put(60,-40){\makebox(0,0){\footnotesize{$\mathrmbfit{tup}_{\mathcal{E}_{1}}(I_{2},s_{2}{\cdot}f)$}}}
\put(60,95){\makebox(0,0){\scriptsize{$k$}}}
\put(60,15){\makebox(0,0){\scriptsize{$\mathrmbfit{tup}(h,f,g)$}}}
\put(-8,40){\makebox(0,0)[r]{\scriptsize{$t_{2}$}}}
\put(130,40){\makebox(0,0)[l]{\scriptsize{$t_{1}$}}}
\put(105,-22){\makebox(0,0)[l]{\scriptsize{$\mathrmbfit{tup}_{\mathcal{E}_{1}}(h)$}}} 
\put(15,-22){\makebox(0,0)[r]{\scriptsize{$\tau_{{\langle{f,g}\rangle}}(I_{2},s_{2})$}}}
\put(0,65){\vector(0,-1){50}}
\put(120,65){\vector(0,-1){50}}
\put(95,80){\vector(-1,0){70}}
\put(95,0){\vector(-1,0){70}}
\put(45,-30){\vector(-3,2){30}}
\put(105,-10){\vector(-3,-2){30}}
\end{picture}
\end{tabular}
&
{\scriptsize\begin{tabular}{r@{\hspace{5pt}}l}
\multicolumn{2}{l}{
$k \cdot t_{2} = t_{1} \cdot (h{\cdot}{(\mbox{-})})\cdot({(\mbox{-})}{\cdot}g)$
and $s_{2} \cdot f = h \cdot s_{1}$}
\\
for all &
$k_{1} \in K_{1}, i_{2} \in I_{2}$
\\
let &
$k_{2} = k(k_{1}) \in K_{1}, i_{1} = h(i_{2}) \in I_{1}$ 
\\
then &
${t_{2}}_{k_{2}} = h \cdot {t_{1}}_{k_{1}} \cdot g$,
$f(s_{2}(i_{2})) = s_{1}(i_{1})$ and
${t_{2}}_{k_{2},i_{2}} = g({t_{1}}_{k_{1},i_{1}})$
\\
hence &
${t_{2}}_{k_{2},i_{2}} \models_{\mathcal{E}_{2}} s_{2}(i_{2})$
iff
${t_{1}}_{k_{1},i_{1}} \models_{\mathcal{E}_{1}} s_{1}(i_{1})$
\end{tabular}}
\\ & \\ & \\ & \\ & \\
\begin{tabular}{c}
\setlength{\unitlength}{0.347pt}
\begin{picture}(180,180)(0,0)
%
\put(0,180){\makebox(0,0){\footnotesize{$I_{2}$}}}
\put(180,180){\makebox(0,0){\footnotesize{$I_{1}$}}}
\put(0,100){\makebox(0,0){\footnotesize{$X_{2}$}}}
\put(180,100){\makebox(0,0){\footnotesize{$X_{1}$}}}
\put(0,0){\makebox(0,0){\footnotesize{$Y_{2}$}}}
\put(180,0){\makebox(0,0){\footnotesize{$Y_{1}$}}}
\put(45,145){\makebox(0,0){\footnotesize{$\hat{I}_{1}$}}}
\qbezier[40](45,125)(40,70)(15,10)\put(15,10){\vector(-1,-3){0}}
\put(40,65){\makebox(0,0)[l]{\scriptsize{${\widehat{(\mbox{-})}}$}}}
\put(-6,140){\makebox(0,0)[r]{\scriptsize{$s_{2}$}}}
\put(22,116){\makebox(0,0)[l]{\scriptsize{$\hat{s}_{1}$}}}
\put(188,140){\makebox(0,0)[l]{\scriptsize{$s_{1}$}}}
\put(-8,50){\makebox(0,0)[r]{\scriptsize{$\models_{2}$}}}
\put(188,50){\makebox(0,0)[l]{\scriptsize{$\models_{1}$}}}
\put(90,194){\makebox(0,0){\scriptsize{$h = \widehat{h}{\cdot}\varepsilon_{f}$}}}
\put(82,166){\makebox(0,0){\scriptsize{$\varepsilon_{f}$}}}
\put(25,168){\makebox(0,0)[l]{\scriptsize{$\widehat{h}$}}}
\put(90,114){\makebox(0,0){\scriptsize{$f$}}}
\put(90,14){\makebox(0,0){\scriptsize{$g$}}}
\put(-48,90){\makebox(0,0)[r]{\scriptsize{${t_{2}}_{k_{2}}$}}}
\put(231,90){\makebox(0,0)[l]{\scriptsize{${t_{1}}_{k_{1}}$}}}
\put(20,180){\vector(1,0){140}}
\put(20,100){\vector(1,0){140}}
\put(160,0){\vector(-1,0){140}}
\put(0,170){\vector(0,-1){60}}
\put(0,90){\vector(0,-1){80}}
\put(180,170){\vector(0,-1){60}}
\put(180,90){\vector(0,-1){80}}
\put(210,90){\oval(30,180)[r]}\put(204,0){\vector(-1,0){0}}
\put(-30,90){\oval(30,180)[l]}\put(-24,0){\vector(1,0){0}}
\put(57,147){\vector(4,1){100}}
\qbezier(10,110)(20,120)(30,130)\put(10,110){\vector(-1,-1){0}}
\qbezier(10,170)(20,160)(30,150)\put(30,150){\vector(1,-1){0}}
\qbezier(55,130)(60,130)(65,130)
\qbezier(65,140)(65,135)(65,130)
\end{picture}
\end{tabular}
&
\begin{tabular}{c}
\setlength{\unitlength}{0.45pt}
\begin{picture}(125,120)(-60,0)
\put(0,0){\begin{picture}(0,0)(0,0)
\put(-10,90){\makebox(0,0){\scriptsize{$\mathcal{T}_{2}$}}}
\put(60,120){\makebox(0,0){\shortstack{\scriptsize{$I_{2}$}\\$\overbrace{\rule{50pt}{0pt}}$}}}
\put(-55,40){\makebox(0,0){\scriptsize{$K_{2}\left\{\rule{0pt}{20pt}\right.$}}}
\put(60,90){\makebox(0,0){\scriptsize{$i_{2}$}}}
\put(-10,40){\makebox(0,0){\scriptsize{$k_{2}$}}}
\put(60,40){\makebox(0,0){\tiny{${t_{2}}_{k_{2},i_{2}}$}}}
\put(-20,78){\line(1,0){140}}
\put(2,0){\line(0,1){100}}
\put(185,170){\makebox(0,0){\scriptsize{$h$}}}
\qbezier(80,140)(185,180)(290,140)\put(290,140){\vector(4,-1){0}}
\end{picture}}
\put(245,0){\begin{picture}(0,0)(0,0)
\put(-10,90){\makebox(0,0){\scriptsize{$\mathcal{T}_{1}$}}}
\put(60,120){\makebox(0,0){\shortstack{\scriptsize{$I_{1}$}\\$\overbrace{\rule{50pt}{0pt}}$}}}
\put(-55,40){\makebox(0,0){\scriptsize{$K_{1}\left\{\rule{0pt}{20pt}\right.$}}}
\put(60,90){\makebox(0,0){\scriptsize{$i_{1}$}}}
\put(-10,40){\makebox(0,0){\scriptsize{$k_{1}$}}}
\put(60,40){\makebox(0,0){\tiny{${t_{1}}_{k_{1},i_{1}}$}}}
\put(-20,78){\line(1,0){140}}
\put(2,0){\line(0,1){100}}
\put(-225,-20){\makebox(0,0){\scriptsize{$k$}}}
\qbezier(-90,25)(-300,-40)(-310,20)\put(-310,20){\vector(-1,3){0}}
\end{picture}}
\end{picture}
\end{tabular}
\\ & \\ & \\
\multicolumn{2}{c}{{\scriptsize$k \cdot t_{2} 
= t_{1} \cdot \mathrmbfit{tup}(h,f,g)
= t_{1} \cdot \mathrmbfit{tup}_{\mathcal{E}_{1}}(h) \cdot \tau_{{\langle{f,g}\rangle}}(I_{2},s_{2})$}}
\\ & \\
\multicolumn{2}{c}{{\scriptsize\begin{tabular}{p{268pt}}
This four-part figure illustrates the condition on table morphisms.
It has been annotated to help guide the understanding.
The condition is symbolically stated in terms of set functions in the line of text just above.
The top left diagram illustrates the condition,
and the bottom left diagram expands on this.
The top right diagram text is more detailed in terms of a source key $k_{1} \in K_{1}$.
Here we see appearance of the infomorphism condition 
\[
g({t_{1}}_{k_{1},i_{1}}) \models_{\mathcal{E}_{2}} s_{2}(i_{2}) 
\;\text{iff}\;
{t_{1}}_{k_{1},i_{1}} \models_{\mathcal{E}_{1}} f(s_{2}(i_{2})).
\]
Finally,
the bottom left figure illustrates the effect of the morphism on source/target tables
$\mathcal{T}_{1}$ and $\mathcal{T}_{2}$.
\end{tabular}}}
\end{tabular}
\end{center}
\caption{Table Morphism}
\label{table:morphism}
\end{figure}
Composition of morphisms is defined component-wise.
Let $\mathrmbf{Tbl}$ denote the category of tables (database relations) 
with the two projections ($\mathrmbfit{sch}$ is called the schema functor)
$\mathrmbf{Sch} \xleftarrow{\mathrmbfit{sch}} \mathrmbf{Tbl}^{\mathrm{op}}\xrightarrow{\mathrmbfit{cls}} \mathrmbf{Cls}$
and the key functor
$\mathrmbf{Tbl} \xrightarrow{\mathrmbfit{key}} \mathrmbf{Set}$.
\begin{center}
\begin{tabular}{c}
\\ \\
\setlength{\unitlength}{0.55pt}
\begin{picture}(320,80)(0,-10)
\put(3,80){\makebox(0,0){\footnotesize{$\mathrmbf{Set}^{\mathrm{op}}$}}}
\put(143,80){\makebox(0,0){\footnotesize{$\mathrmbf{Tbl}^{\mathrm{op}}$}}}
\put(280,80){\makebox(0,0){\footnotesize{$\mathrmbf{Cls}$}}}
\put(0,0){\makebox(0,0){\footnotesize{$\mathrmbf{1}$}}}
\put(140,0){\makebox(0,0){\footnotesize{$\mathrmbf{Sch}$}}}
\put(280,0){\makebox(0,0){\footnotesize{$\mathrmbf{Set}$}}}
\put(210,40){\makebox(0,0){\footnotesize{$\mathrmbf{Sdsgn}$}}}
\put(80,92){\makebox(0,0){\scriptsize{$\mathrmbfit{key}^{\mathrm{op}}$}}}
\put(210,92){\makebox(0,0){\scriptsize{$\mathrmbfit{cls}$}}}
\put(210,-12){\makebox(0,0){\scriptsize{$\mathrmbfit{typ}$}}}
\put(132,40){\makebox(0,0)[r]{\scriptsize{$\mathrmbfit{sch}$}}}
\put(288,40){\makebox(0,0)[l]{\scriptsize{$\mathrmbfit{typ}$}}}
\put(140,115){\makebox(0,0){\tiny{${\langle{\mathcal{S}_{2},\mathcal{E}_{2},K_{2},t_{2}}\rangle} 
\xleftarrow{{\langle{h,f,g,k}\rangle}}{\langle{\mathcal{S}_{1},\mathcal{E}_{1},K_{1},t_{1}}\rangle}$}}}
\put(-60,80){\makebox(0,0)[r]{\tiny{$K_{2}\xleftarrow{k}K_{1}$}}}
\put(310,80){\makebox(0,0)[l]{\tiny{${\langle{X_{2},Y_{2},\models_{\mathcal{E}_{2}}}\rangle} \xrightarrow{{\langle{f,g}\rangle}}{\langle{X_{1},Y_{1},\models_{\mathcal{E}_{1}}}\rangle}$}}}
\put(140,-25){\makebox(0,0){\tiny{${\langle{X_{2},I_{2},s_{2}}\rangle}
\xrightarrow{{\langle{h,f}\rangle}}{\langle{X_{1},I_{1},s_{1}}\rangle}$}}}
\put(360,0){\makebox(0,0)[l]{\tiny{$X_{2}\xrightarrow{f}X_{1}$}}}
\put(115,80){\vector(-1,0){90}}
\put(165,80){\vector(1,0){90}}
\put(115,0){\vector(-1,0){100}}
\put(165,0){\vector(1,0){90}}
\put(0,65){\vector(0,-1){50}}
\put(140,65){\vector(0,-1){50}}
\put(280,65){\vector(0,-1){50}}
\put(155,70){\vector(2,-1){36}}
\put(190,30){\vector(-2,-1){36}}
\put(230,50){\vector(2,1){36}}
\qbezier(230,20)(235,20)(240,20)
\qbezier(240,20)(240,25)(240,30)
\end{picture}
\\ \\
\end{tabular}
\end{center}
We can have three indexing categories for tables:
classifications, schema or semidesignations.
Each has their uses:
classification indexing proves that the category of tables is complete
(and the fibers help explain database fibers),
semidesignation indexing proves that the category of tables is cocomplete, and
schema indexing follows the true formal-semantics distinction.

\subsection{Classification Indexed Category.}

For any fixed classification $\mathcal{E}$,
the $\mathcal{E}^{\mathrm{th}}$-fiber category with respect to the classification functor
$\mathrmbf{Tbl}^{\mathrm{op}} \xrightarrow{\mathrmbfit{cls}} \mathrmbf{Cls}$,
called the category of $\mathcal{E}$-tables,
is the comma category 
associated with the tuple functor
$\mathrmbfit{tup}_{\mathcal{E}} : (\mathrmbf{Set}{\downarrow}X)^{\mathrm{op}} \rightarrow \mathrmbf{Set}$.
{\footnotesize\[
\mathrmbf{Set} \xleftarrow{\mathrmbfit{key}_{\mathcal{E}}}\mathrmbf{Tbl}(\mathcal{E}) = {(\mathrmbf{Set}{\downarrow}\mathrmbfit{tup}_{\mathcal{E}})}\xrightarrow{\mathrmbfit{sign}_{\mathcal{E}}^{\mathrm{op}}} {(\mathrmbf{Set}{\downarrow}X)}^{\mathrm{op}}.
\]\normalsize}
\hspace{-6pt}
It has key and signature projection functors, 
an equivalent natural transformation
$\tau : 
\mathrmbfit{key}_{\mathcal{E}} \Rightarrow 
\mathrmbfit{sign}^{\mathrm{op}}_{\mathcal{E}} \circ \mathrmbfit{tup}_{\mathcal{E}}$, and
is described as follows.
A \emph{fiber object} $\mathcal{T} \in \mathrmbf{Tbl}(\mathcal{E})$,
or an $\mathcal{E}$-table (database $\mathcal{E}$-relation),
is any table $\mathcal{T} \in \mathrmbf{Tbl}$ with entity classification $\mathrmbfit{cls}(\mathcal{T}) = \mathcal{E}$
and tuple (content) function $t : K \rightarrow \mathrmbfit{tup}_{\mathcal{E}}(I,s)$
mapping each key (abstract tuple) to a (concrete) $\mathcal{E}$-tuple in the extent of ${\langle{I,s}\rangle}$.
A \emph{fiber morphism} in $\mathrmbf{Tbl}(\mathcal{E})$ 
is any table morphism
${\langle{h,k}\rangle} :
\mathcal{T} = {\langle{\mathcal{S},\mathcal{E},K,t}\rangle} \leftarrow
{\langle{\widetilde{\mathcal{S}},\mathcal{E},\widetilde{K},\widetilde{t}}\rangle} = \widetilde{\mathcal{T}}$
in $\mathrmbf{Tbl}$ 
with identity infomorphism 
$\mathrmit{id}_{\mathcal{E}} = {\langle{\mathrmit{id}_{X},\mathrmit{id}_{Y}}\rangle}$.
It consists of  
a signature morphism $h : {\langle{I,s}\rangle} \rightarrow {\langle{\widetilde{I},\widetilde{s}}\rangle}$ and
a key function $k : \widetilde{K} \rightarrow K$,
which satisfy the condition
$k \cdot t = \widetilde{t} \cdot \mathrmbfit{tup}_{\mathcal{E}}(h)$.
\begin{proposition}
There is an indexed category of tables
$\mathrmbfit{tbl} : \mathrmbf{Cls}^{\mathrm{op}} \rightarrow \mathrmbf{Cat}$  
from (the opposite of) the category of classifications and infomorphisms
to the category of categories and functors.~\footnote{The table indexing functor is the composition
$\mathrmbfit{tbl} = \mathrmbfit{tup}^{\mathrm{op}} \circ \mathrmbfit{comma} \circ {(\mbox{-})}^{\mathrm{op}}:
\mathrmbf{Cls}^{\mathrm{op}} 
\rightarrow \left(\mathrmbf{Adj}{\,\Downarrow\,}\mathrmbf{Set}\right)^{\mathrm{op}} 
\rightarrow \mathrmbf{Cat} \rightarrow \mathrmbf{Cat}$
of a tuple functor
$\mathrmbf{Cls} \xrightarrow{\mathrmbfit{tup}} \left(\mathrmbf{Adj}{\,\Downarrow\,}\mathrmbf{Set}\right)$
and a comma category functor
$\left(\mathrmbf{Adj}{\,\Downarrow\,}\mathrmbf{Set}\right)^{\mathrm{op}} \xrightarrow{\mathrmbfit{comma}} \mathrmbf{Cat}$.}
The (opposite of the) fibered category corresponding to this (its Grothendieck construction)
is isomorphic to the category of tables with the classification functor projection
$\mathrmbf{Tbl}^{\mathrm{op}} \xrightarrow{\mathrmbfit{cls}} \mathrmbf{Cls}$.
\end{proposition}
\begin{proposition}
The category of $\mathcal{E}$-tables 
$\mathrmbf{Tbl}({\mathcal{E}}) = (\mathrmbf{Set}{\downarrow}\mathrmbfit{tup}_{\mathcal{E}})$ 
is complete,
its key projection functor 
$\mathrmbf{Tbl}(\mathcal{E}) \xrightarrow{\mathrmbfit{key}_{\mathcal{E}}} \mathrmbf{Set}$ 
is continuous and
its signature projection functor 
$\mathrmbf{Tbl}(\mathcal{E})^{\mathrm{op}} \xrightarrow{\mathrmbfit{sign}_{\mathcal{E}}} (\mathrmbf{Set}{\downarrow}X)$ 
is cocontinuous.
\end{proposition}
\begin{proof}
(We have both an abstract and a useful concrete proof, 
but only have room for the former.)
The category $(\mathrmbf{Set}{\downarrow}X)^{\mathrm{op}}$ is complete,
since $(\mathrmbf{Set}{\downarrow}X)$ is cocomplete.
The tuple functor $\mathrmbfit{tup}_{\mathcal{E}}$ is continuous.
\footnote{If $\mathcal{E}$ and $\mathcal{B}$ are complete, and 
both $T : \mathcal{E} \rightarrow \mathcal{C}$ and $S : \mathcal{B} \rightarrow \mathcal{C}$ are continuous functors, 
then the comma category $(T \downarrow S)$ is also complete and 
the projection functors 
$(T{\downarrow}S) \rightarrow \mathcal{E}$ and $(T{\downarrow}S) \rightarrow \mathcal{B}$ 
are limit preserving.}
\end{proof}
The category of $\mathcal{E}$-tables $\mathrmbf{Tbl}({\mathcal{E}})$
is the semantical domain for a relational database $\mathcal{D}$ with entity classification $\mathcal{E}$.
Completeness of $\mathrmbf{Tbl}({\mathcal{E}})$ means that, not just binary, 
but database joins over arbitrary diagrams of tables of $\mathcal{D}$ can be computed.  
\begin{proposition}
For any infomorphism
${\langle{f,g}\rangle} : \mathcal{E}_{2} \rightarrow \mathcal{E}_{1}$,
the 
table fiber functor
$\mathrmbfit{tbl}_{{\langle{f,g}\rangle}} :
\mathrmbf{Tbl}(\mathcal{E}_{1}) 
\rightarrow 
\mathrmbf{Tbl}(\mathcal{E}_{2})$
is continuous.
\end{proposition}
Continuity of $\mathrmbfit{tbl}_{{\langle{f,g}\rangle}}$
means that database joins are preserved:
database joins of $\mathcal{E}_{1}$-tables are mapped to
database joins of $\mathcal{E}_{2}$-tables.
\begin{theorem}
The category of tables $\mathrmbf{Tbl}$ is complete.
\end{theorem}
\begin{proof}
The indexing category $\mathrmbf{Cls}$ is complete,
the fiber category $\mathrmbfit{tbl}(\mathcal{E})$
is complete for each classification $\mathcal{E}$, and
the fiber functor
$\mathrmbfit{tbl}_{{\langle{f,g}\rangle}} : \mathrmbfit{tbl}(\mathcal{E}_{1}) \rightarrow \mathrmbfit{tbl}(\mathcal{E}_{2})$
is continuous for each infomorphism ${\langle{f,g}\rangle} : \mathcal{E}_{2} \rightarrow \mathcal{E}_{1}$.
Hence,
this is an application of a theorem of Tarlecki, Burstall and Goguen~\cite{tarlecki:burstall:goguen:91}.
~\footnote{If $\mathrmbf{C} : \mathrmbf{I}^{\mathrm{op}} \rightarrow \mathrmbf{Cat}$ 
is an indexed category such that
%
%
$\mathrmbf{I}$ is complete,
%
$\mathrmbf{C}_{i}$ is complete for all indices $i \in \mathrmbf{I}$, and
%
$\mathrmbf{C}_{\sigma} : \mathrmbf{C}_{j} \rightarrow \mathrmbf{C}_{i}$ is continuous 
for all index morphisms $\sigma : i \rightarrow j$,
%
%
then $\mathrmbf{Gr}(\mathrmbf{C})$ is complete.}
\end{proof}

\subsection{Schema Indexed Category.}

For any fixed (simple) schema $\mathcal{S} = {\langle{X,I,s}\rangle}$,
the $\mathcal{S}^{\mathrm{th}}$-fiber category $\mathrmbf{Tbl}(\mathcal{S})$
with respect to the schema functor
$\mathrmbf{Tbl}^{\mathrm{op}} \xrightarrow{\mathrmbfit{sch}} \mathrmbf{Sch}$,
called the category of $\mathcal{S}$-tables,
is the comma category 
with key and $X$-classification projection functors 
\footnote{The tuple functor
$\mathrmbfit{tup}_{\mathcal{S}} : \mathrmbf{Cls}(X)^{\mathrm{op}} \rightarrow \mathrmbf{Set}$
maps an $X$-classification
$\mathcal{E} = {\langle{X,Y,\models}\rangle}$
to the tuple set 
$\mathrmbfit{tup}_{\mathcal{S}}(Y,\models) = \mathrmbfit{tup}_{\mathcal{E}}(I,S)$
and maps an $X$-infomorphism
${\langle{\mathrmit{1}_{X},g}\rangle} : 
\mathcal{E}_{2} = {\langle{X,Y_{2},\models_{2}}\rangle} \rightleftarrows 
{\langle{X,Y_{1},\models_{1}}\rangle} = g^{-1}(\mathcal{E}_{2}) = \mathcal{E}_{1}$
with instance function $g : Y_{1} \rightarrow Y_{2}$
to the tuple function
$\mathrmbfit{tup}_{\mathcal{S}}(g) = \tau_{{\langle{\mathrmit{1}_{X},g}\rangle}}(I,s) = {(\mbox{-})}{\cdot}{g} : 
\mathrmbfit{tup}_{\mathcal{S}}(Y_{1},\models_{1}) = \mathrmbfit{tup}_{\mathcal{E}_{1}}(I,S) \rightarrow
\mathrmbfit{tup}_{\mathcal{E}_{2}}(I,S) = \mathrmbfit{tup}_{\mathcal{S}}(Y_{2},\models_{2})$.}
{\footnotesize\[
\mathrmbf{Set} \xleftarrow{\mathrmbfit{key}_{\mathcal{S}}}\mathrmbf{Tbl}(\mathcal{S}) = {(\mathrmbf{Set}{\downarrow}\mathrmbfit{tup}_{\mathcal{S}})}\xrightarrow{\mathrmbfit{cls}_{\mathcal{S}}^{\mathrm{op}}} {\mathrmbf{Cls}(X)}^{\mathrm{op}}.
\]\normalsize}
\hspace{-6pt}
It is described as follows.
A \emph{fiber object} $\mathcal{T} \in \mathrmbf{Tbl}(\mathcal{S})$,
or an $\mathcal{S}$-table (database $\mathcal{S}$-relation),
is any table $\mathcal{T} \in \mathrmbf{Tbl}$ with (simple) schema $\mathrmbfit{sch}(\mathcal{T}) = \mathcal{S}$.
A \emph{fiber morphism} in $\mathrmbf{Tbl}(\mathcal{S})$
is any table morphism
${\langle{g,k}\rangle} :
\mathcal{T} = {\langle{\mathcal{S},\mathcal{E},K,t}\rangle} \leftarrow
{\langle{\mathcal{S},\widetilde{\mathcal{E}},\widetilde{K},\widetilde{t}}\rangle} = \widetilde{\mathcal{T}}$
in $\mathrmbf{Tbl}$ 
with identity (simple) schema morphism 
$\mathrmit{id}_{\mathcal{S}} = {\langle{\mathrmit{id}_{X},\mathrmit{id}_{I}}\rangle}$.
It consists of 
a universe (entity instance) function $g : \widetilde{Y} \rightarrow Y$
defining
an entity infomorphism 
${\langle{\mathrmit{1}_{X},g}\rangle} : 
\mathcal{E} = {\langle{X,Y,\models_{\mathcal{E}}}\rangle} \rightleftarrows 
{\langle{X,\widetilde{Y},\widetilde{\models}}\rangle} = g^{-1}(\mathcal{E}) = \widetilde{\mathcal{E}}$
and hence the presheaf morphism
${\langle{(\mathrmbf{Set}{\downarrow}X),\mathrmbfit{tup}_{\mathcal{E}}}\rangle} 
\xrightarrow{{\langle{\mathrmit{1},\tau_{{\langle{\mathrmit{1}_{X},g}\rangle}}}\rangle}}
{\langle{(\mathrmbf{Set}{\downarrow}X),\mathrmbfit{tup}_{\widetilde{\mathcal{E}}}}\rangle}$
with tuple natural transformation
$\tau_{{\langle{\mathrmit{1}_{X},g}\rangle}} :
\mathrmbfit{tup}_{\widetilde{\mathcal{E}}} \Rightarrow \mathrmbfit{tup}_{\mathcal{E}}$,
and a key function $k : \widetilde{K} \rightarrow K$,
which satisfy the condition
$k \cdot t = \widetilde{t} \cdot \mathrmbfit{tup}(g)$.
\footnote{The components determining variance between 
$\mathcal{T} = {\langle{\mathcal{S},\mathcal{E},K,t}\rangle}$
and
$\widetilde{\mathcal{T}} = {\langle{\mathcal{S},\widetilde{\mathcal{E}},\widetilde{K},\widetilde{t}}\rangle}$
are the entity instance function (varying instances and their incidence or classification relations)
and
the key function
(varying the set of keys and the tuple natural transformations).}
%
%
\begin{proposition}
There is an indexed category of tables 
$\mathrmbfit{tbl} : \mathrmbf{Sch}^{\mathrm{op}} \rightarrow  \mathrmbf{Cat}$,  
whose Grothendieck construction (fibered category)
is (the opposite of) the category of tables with the schema functor projection
$\mathrmbf{Tbl}^{\mathrm{op}} \xrightarrow{\mathrmbfit{sch}} \mathrmbf{Sch}$.
\end{proposition}
%

\subsection{Semidesignation Indexed Category.}

A semidesignation
$\mathcal{S} = {\langle{I,s,\mathcal{E}}\rangle}$,
consists of a schema
${\langle{X,I,s}\rangle}$,
and an entity classification $\mathcal{E} = {\langle{X,Y,\models_{\mathcal{E}}}\rangle}$ 
with a common (entity) type set component $X$.
A semidesignation morphism
${\langle{h,f,g}\rangle} : \mathcal{S}_{2} \rightarrow \mathcal{S}_{1}$
consists of 
a schema morphism
${\langle{h,f}\rangle} : \mathcal{S}_{2} \rightarrow \mathcal{S}_{1}$
and 
an entity infomorphism 
${\langle{f,g}\rangle} : 
\mathcal{E}_{2} = {\langle{X_{2},Y_{2},\models_{\mathcal{E}_{2}}}\rangle} \rightleftarrows 
{\langle{X_{1},Y_{1},\models_{\mathcal{E}_{1}}}\rangle} = \mathcal{E}_{1}$
with a common (entity) type function component $f : X_{2} \rightarrow X_{1}$.
For any semidesignation
$\mathcal{S} = {\langle{I,s,\mathcal{E}}\rangle}$,
the set of tuples of $\mathcal{S}$ is 
$\mathrmbfit{tup}(\mathcal{S}) = \mathrmbfit{tup}_{\mathcal{E}}(I,s)$,
the set of $\mathcal{E}$-tuples in the extent of ${\langle{I,s}\rangle}$.
\begin{lemma}
Any semidesignation morphism
${\langle{h,f,g}\rangle} : \mathcal{S}_{2} \rightarrow \mathcal{S}_{1}$
defines a tuple function 
$\mathrmbfit{tup}(h,f,g) : 
\mathrmbfit{tup}(\mathcal{S}_{1}) = \mathrmbfit{tup}_{\mathcal{E}_{1}}(I_{1},s_{1}) \rightarrow 
\mathrmbfit{tup}_{\mathcal{E}_{2}}(I_{2},s_{2}) = \mathrmbfit{tup}(\mathcal{S}_{2})$.
Hence, there is a tuple functor
$\mathrmbfit{tup} : \mathrmbf{Sdsgn}^{\mathrm{op}} \rightarrow \mathrmbf{Set}$.
\end{lemma}
\begin{proposition}
The category of tables is the comma category 
{\footnotesize\[
\mathrmbf{Set} \xleftarrow{\mathrmbfit{key}} \mathrmbf{Tbl}=(\mathrmbf{Set}{\downarrow}\mathrmbfit{tup})
\xrightarrow{\mathrmbfit{sdsgn}} \mathrmbf{Sdsgn}^{\mathrm{op}}
\xrightarrow{\mathrmbfit{cls}^{\mathrm{op}}} \mathrmbf{Cls}^{\mathrm{op}}
\]\normalsize}
\hspace{-6pt}
associated with the tuple functor
$\mathrmbfit{tup} : \mathrmbf{Sdsgn}^{\mathrm{op}} \rightarrow \mathrmbf{Set}$.
The category of tables is cocomplete.
\end{proposition}
\begin{proof}
The opposite category of semidesignations $\mathrmbf{Sdsgn}^{\mathrm{op}}$ is cocomplete,
since $\mathrmbf{Sdsgn}$ is complete.
\footnote{If $\mathcal{A}$ and $\mathcal{B}$ are cocomplete, 
$T : \mathcal{A} \rightarrow \mathcal{C}$ is a cocontinuous functor, and 
$S : \mathcal{B} \rightarrow \mathcal{C}$ is any functor (not necessarily cocontinuous), 
then the comma category $(T{\downarrow}S)$ will also be cocomplete.}
\newline
\end{proof}

\section{Relational Databases}\label{relational:databases}

A \emph{relational database} 
$\mathcal{D} = {\langle{\mathcal{S},\mathcal{E},\mathrmbfit{K},\tau}\rangle}$
is a naturally connected diagram of tables.
It has an underlying relational database schema 
$\mathcal{S} = {\langle{\mathrmbf{R},X,\mathrmbfit{S}}\rangle}$ 
\footnote{A relational database schema
$\mathcal{S} = {\langle{\mathrmbf{R},X,\mathrmbfit{S}}\rangle}$
consists of 
a category of relation symbols $\mathrmbf{R}$, 
a set of entity types $X$, and
a signature functor $\mathrmbfit{S} : \mathrmbf{R} \rightarrow (\mathrmbf{Set}{\downarrow}X)$.
Any relational database schema
$\mathcal{S} = {\langle{\mathrmbf{R},X,\mathrmbfit{S}}\rangle}$
with colimit reference schema ${\langle{I,s}\rangle} = \coprod\mathrmbfit{S}$,
defines a type language $\mathrmbfit{lang}(\mathrmbfit{S})$
in the Information Flow Framework~\cite{iff}
with reference component ${\langle{X,I,s}\rangle}$
and signature component ${\langle{\mathrmbf{R},\partial}\rangle}$.
We regard 
the colimit $X$-signature ${\langle{I,s}\rangle} = \coprod\mathrmbfit{S}$ 
to be a reference schema ${\langle{X,I,s}\rangle}$
with reference (sort) function $I \xrightarrow{s} X$
from a universal set of variables $I$ to the type set $X$.
For any relation symbol $r \in \mathrmbf{R}$,
the colimit injection
$\mathrmbfit{S}(r) = {\langle{I_{r},s_{r}}\rangle} \xrightarrow{\iota_{r}} {\langle{I,s}\rangle}$,
whose condition $s_{r} = \iota_{r} \cdot s$ expresses the $s$-alignment of $s_{r}$ via $\iota_{r}$,
states that 
the signature ${\langle{I_{r},s_{r}}\rangle}$
is below (at least as general as) 
the colimit signature ${\langle{I,s}\rangle}$.
The signature functor $\partial$ factors
$\mathrmbfit{S} = \partial \circ \mathrmbfit{inc} : \mathrmbf{R} \rightarrow (\mathrmbf{Set}{\downarrow}X)$
through 
$\mathrmbfit{sign}(I,s) \subseteq (\mathrmbf{Set}{\downarrow}X)$,
the subcategory of $X$-signatures below ${\langle{I,s}\rangle}$.}
with
a category of relation types (symbols) $\mathrmbf{R}$ linked by generalization-specialization, 
a set of entity types $X$, and
a signature functor
$\mathrmbfit{S} : \mathrmbf{R} \rightarrow (\mathrmbf{Set}{\downarrow}X)$,
an entity classification $\mathcal{E} = {\langle{X,Y,\models_{\mathcal{E}}}\rangle}$ 
with a common (entity) type set component $X$ and a local universe of entity instances $Y$,
a key functor $\mathrmbfit{K} : \mathrmbf{R}^{\mathrm{op}} \rightarrow \mathrmbf{Set}$, and
a tuple natural transformation
$\tau : \mathrmbfit{K} \Rightarrow \mathrmbfit{S}^{\mathrm{op}} \circ \mathrmbfit{tup}_{\mathcal{E}}$.
Equivalently,
it consists of a table functor 
$\mathrmbfit{T} : \mathrmbf{R}^{\mathrm{op}} \rightarrow (\mathrmbf{Set}{\downarrow}\mathrmbfit{tup}_{\mathcal{E}})$,
where $\tau = \mathrmbfit{T}\tau_{\mathcal{E}}$ and
$\tau_{\mathcal{E}} : 
\mathrmbfit{key}_{\mathcal{E}} \Rightarrow \mathrmbfit{sign}^{\mathrm{op}}_{\mathcal{E}} \circ \mathrmbfit{tup}_{\mathcal{E}}$
is the tuple natural transformation that is an integral component of the comma category 
$(\mathrmbf{Set}{\downarrow}\mathrmbfit{tup}_{\mathcal{E}})$.
Here are some examples of relational databases.
\begin{description}
\item[Table.]
A table (database relation) 
$\mathcal{T} = {\langle{\mathcal{S},\mathcal{E},K,t}\rangle}$
with
entity classification $\mathcal{E} = {\langle{X,Y,\models_{\mathcal{E}}}\rangle}$, 
schema $\mathcal{S} = {\langle{X,I,s}\rangle}$, 
tuple set $K$, and
tuple function $t : K \rightarrow \mathrmbfit{tup}_{\mathcal{E}}(I,s)$,
is a one-object relational database
with the same entity classification,
the terminal category of relation types (symbols) $\mathrmbf{1}=\{\ast\}$,
the signature functor with single $X$-signature
$\mathrmbf{1} \xrightarrow{{\langle{I,s}\rangle}} (\mathrmbf{Set}{\downarrow}X)$,
the key functor with single key set $\mathrmbf{1}^{\mathrm{op}}=\mathrmbf{1} \xrightarrow{K} \mathrmbf{Set}$,
the tuple natural transformation with single component tuple function
$t : K \rightarrow \mathrmbfit{tup}_{\mathcal{E}}(I,s)$, and
the table functor with single $\mathcal{E}$-table 
$\mathrmbf{1}^{\mathrm{op}}=\mathrmbf{1} \xrightarrow{\mathcal{T}} (\mathrmbf{Set}{\downarrow}\mathrmbfit{tup}_{\mathcal{E}})$.
\item[Classification.] 
A classification $\mathcal{E} = {\langle{X,Y,\models_{\mathcal{E}}}\rangle}$
is a relational database
$\mathrmbfit{db}(\mathcal{E}) = {\langle{\geq_{\mathcal{E}},{\uparrow}_{\mathcal{E}},\mathcal{E},\mathrmbfit{ext}_{\mathcal{E}},\tau_{\mathcal{E}}}\rangle}$,
where the entity classification is itself
($\mathrmbfit{db} \circ \mathrmbfit{cls} = \mathrmbfit{id}_{\mathrmbf{Cls}}$).
~\footnote{Since any preorder $\mathcal{P} = {\langle{P,\leq}\rangle}$
is a classification $\mathcal{P} = {\langle{P,P,\leq}\rangle}$,
a preorder is a relational database.}
The additional components are described as follows. 
The category of relation types (symbols) 
is the reverse conceptual preorder (generalization-specialization order) on entity types 
${\langle{X,\geq_{\mathcal{E}}}\rangle}$ 
with $x' \geq_{\mathcal{E}} x$ when $\mathrmbfit{ext}_{\mathcal{E}}(x') \supseteq \mathrmbfit{ext}_{\mathcal{E}}(x)$;
\makebox(0,0){
\setlength{\unitlength}{0.45pt}
\begin{picture}(0,0)(0,0)
\put(-300,-100){\begin{picture}(0,0)(0,0)
\put(-10,90){\makebox(0,0){\footnotesize{$\mathcal{E}$}}}
\put(60,120){\makebox(0,0){\shortstack{\scriptsize{$X$}\\$\overbrace{\rule{50pt}{0pt}}$}}}
\put(-40,40){\makebox(0,0){\scriptsize{$Y\left\{\rule{0pt}{20pt}\right.$}}}
\put(60,90){\makebox(0,0){\scriptsize{$x$}}}
\put(-10,40){\makebox(0,0){\scriptsize{$y$}}}
\put(60,40){\makebox(0,0){\tiny{$\times$}}}
\put(100,93){\makebox(0,0){\scriptsize{$x'$}}}
\put(-7,10){\makebox(0,0){\scriptsize{$y'$}}}
\put(100,10){\makebox(0,0){\tiny{$\times$}}}
\qbezier[10](65,40)(93,40)(120,40)
\qbezier[7](60,35)(60,17)(60,0)
\put(80,90){\makebox(0,0){\tiny{$\leq$}}}
\put(-10,25){\makebox(0,0){\tiny{$\geq$}}}
\put(-20,78){\line(1,0){140}}
\put(2,0){\line(0,1){100}}
\put(53,-30){\makebox(0,0){\scriptsize{$y' \leq_{\mathcal{E}} y \models_{\mathcal{E}} x \leq_{\mathcal{E}} x'$}}}
\put(25,-80){\shortstack{
\scriptsize{$\mathrmbfit{ext}_{\mathcal{E}}(x) \subseteq \mathrmbfit{ext}_{\mathcal{E}}(x')$} \\
\scriptsize{${\uparrow}x \supseteq {\uparrow}x'$}}}
\end{picture}}
\end{picture}
}
that is, 
when $x'$ is at least as general as $x$.
The signature functor is the principle filter operator 
${\langle{X,\geq_{\mathcal{E}}}\rangle} \xrightarrow{{\uparrow}_{\mathcal{E}}} (\mathrmbf{Set}{\downarrow}X)$;
it maps an entity type $x \in X$ 
to the $X$-signature ${\langle{{\uparrow}x,\mathrmbfit{inc}_{x}}\rangle}$,
where ${\uparrow}x = \{ x' \in X \mid x' \geq_{\mathcal{E}} x \}$ is the principle filter of $x$
and $\mathrmbfit{inc}_{x} : {\uparrow}x \hookrightarrow X$ is inclusion; and
it maps an ordering $x' \geq_{\mathcal{E}} x$ with ${\uparrow}x' \subseteq {\uparrow}x$
to the arity inclusion function $\mathrmbfit{inc}_{x',x} : {\uparrow}x' \hookrightarrow {\uparrow}x$.
For any an entity type $x \in X$, 
the $\mathcal{E}$-tuple functor applied to the $X$-signature ${\langle{{\uparrow}x,\mathrmbfit{inc}_{x}}\rangle}$
is the tuple set
$\mathrmbfit{tup}_{\mathcal{E}}({\uparrow}x,\mathrmbfit{inc}_{x})
= \{ {\uparrow}x \xrightarrow{t} Y \mid t(x') \models_{\mathcal{E}} x', \forall x' \geq_{\mathcal{E}} x \}$.
The key functor is the extent operator
${\langle{X,\leq_{\mathcal{E}}}\rangle} \xrightarrow{\mathrmbfit{ext}_{\mathcal{E}}} \mathrmbf{Set}$,
which maps an entity type $x \in X$ 
to its extent $\mathrmbfit{ext}_{\mathcal{E}}(x) = \{ y \in Y \mid y \models_{\mathcal{E}} x \}$ 
and maps the ordering $x \leq_{\mathcal{E}} x'$ 
with $\mathrmbfit{ext}_{\mathcal{E}}(x) \subseteq \mathrmbfit{ext}_{\mathcal{E}}(x')$
to the extent inclusion function 
$\mathrmbfit{ext}_{\mathcal{E}}(x,x') : 
\mathrmbfit{ext}_{\mathcal{E}}(x) \hookrightarrow \mathrmbfit{ext}_{\mathcal{E}}(x')$.
For any entity type $x \in X$,
the tuple function
$\tau_{\mathcal{E}}(x) : \mathrmbfit{ext}_{\mathcal{E}}(x) \rightarrow 
\mathrmbfit{tup}_{\mathcal{E}}({\uparrow}x,\mathrmbfit{inc}_{x})$
maps an entity instance
$y \in \mathrmbfit{ext}_{\mathcal{E}}(x)$
to the constant tuple ${\uparrow}x \xrightarrow{y} Y$.
This defines a natural transformation
$\tau_{\mathcal{E}} 
: \mathrmbfit{ext}_{\mathcal{E}} \Rightarrow 
{({\uparrow}_{\mathcal{E}})}^{\mathrm{op}} \circ \mathrmbfit{tup}_{\mathcal{E}}$.
\item[IFF Structure.]
Using the key functor,
we can define the relation classification 
$\mathcal{R} = {\langle{R,K,\models_{\mathcal{R}}}\rangle}$ 
with 
type set $R = \mathrmbfit{obj}(\mathrmbf{R})$,
instance set $K = \bigcup_{r \in \mathrmbf{R}} \mathrmbfit{K}(r)$, and 
incidence with $k \models_{\mathcal{R}} r$ when $k \in \mathrmbfit{K}(r)$
for key $k \in K = \bigcup_{r \in \mathrmbf{R}} \mathrmbfit{K}(r)$ 
and relation symbol $r \in R = \mathrmbfit{obj}(\mathrmbf{R})$.
The elements (keys) in $K$ are called abstract tuples in the Information Flow Framework~\cite{iff}. 
Any relational database
$\mathcal{D} = {\langle{\mathcal{S},\mathcal{E},\mathrmbfit{K},\tau}\rangle}$
determines an structure (model) $\mathrmbfit{struc}(\mathrmbfit{D})$
in the Information Flow Framework~\cite{iff}
with type language $\mathrmbfit{lang}(\mathrmbfit{S})$,
entity classification $\mathcal{E}$,
semidesignation ${\langle{I,s,\mathcal{E}}\rangle}$ and
relation classification $\mathcal{R}$.
This is an adjoint situation:
any IFF structure determines a relational database.
\item[Unified Database.]
A unified database is a special case of a database,
whose relation classification coincides with its entity classification $\mathcal{R} = \mathcal{E}$.
Unified databases allow the introduction of foreign keys.
In fact,
columns are either the single primary key or a foreign key.
The entries in a column are keys of the type of the column.
Actual datatypes, such as strings or numbers,
can be regarded as primary keys of themselves.
Conversely,
we can think of any relational table with a single column,
one whose schema is of the form $1 \xrightarrow{x} X$,
to be a set of entities.
\end{description}
Any relational database schema $\mathcal{S} = {\langle{\mathrmbf{R},X,\mathrmbfit{S}}\rangle}$ 
in unified form ($R = \mathrmbfit{obj}(\mathrmbf{R}) = X$)
has an associated sketch. 
Define the arity functor 
$\mathrmbfit{A} = \mathrmbfit{S} \circ \mathrmbfit{set}_{X} 
: \mathrmbf{R} \rightarrow (\mathrmbf{Set}{\downarrow}X) \rightarrow \mathrmbf{Set}$,
Let $\int\mathrmbfit{A} \xrightarrow{\mathrmbfit{pr}} \mathrmbf{R}$ 
denote the Grothendieck construction of $\mathrmbfit{A}$
with object set 
$\coprod_{r \in R} \mathrmbfit{A}(r) = \{ (r,i) \mid r \in R, i \in I, {\langle{I,s}\rangle} = \mathrmbfit{S}(r) \}$
and morphisms $(r',i') \xrightarrow{p} (r,i)$
for $\mathrmbf{R}$-constraints $r' \xrightarrow{p} r$.
The graph $\mathrmbfit{gph}(\mathcal{S})$ of the sketch
has node set $R$ and edges $(r,i) \in \int\mathrmbfit{A}$
with source and target $r \xrightarrow{(r,i)} s(i)$.
This graph is actually 2-dimensional, 
given the $\mathrmbf{R}$-constraints.
The sketch specifies a cone for the signature
of each relation type $r \in \mathrmbf{R}$
and constraints for the commuting diagrams in $\mathrmbf{R}$.
Any relational database $\mathcal{D} = {\langle{\mathcal{S},\mathcal{E},\mathrmbfit{K},\tau}\rangle}$ 
in unified form
has an associated sketch interpretation 
$\mathrmbfit{gph}(\mathcal{S})^{\mathrm{op}} \rightarrow \mathrmbf{Set}$.
The interpretation maps a node (relation type) $r \in \mathrmbf{R}$ to
$\mathrmbfit{K}(r)$ the set of keys of $r$ and
maps an edge $r \xrightarrow{(r,i)} s(i)$ 
to the map 
$\mathrmbfit{K}(r) \rightarrow \mathrmbfit{K}(s(i)) : k \mapsto \tau_{r}(k)(i)$,
where $\tau_{r}(k) \in \mathrmbfit{tup}_{\mathcal{E}}(I,s)$.
This also is 2-dimensional.


A \emph{relational database morphism}
${\langle{\mathrmbfit{F},\theta,f,g,\kappa}\rangle} :
\mathcal{D}_{2} = {\langle{\mathcal{S}_{2},\mathcal{E}_{2},\mathrmbfit{K}_{2},\tau_{2}}\rangle} \rightarrow
{\langle{\mathcal{S}_{1},\mathcal{E}_{1},\mathrmbfit{K}_{1},\tau_{1}}\rangle} = \mathcal{D}_{1}$
consists of 
a relational database schema morphism
${\langle{\mathrmbfit{F},\theta,f}\rangle} : \mathcal{S}_{2} \rightarrow \mathcal{S}_{1}$
\footnote{A relational database schema morphism
${\langle{\mathrmbfit{F},\theta,f}\rangle} : 
\mathcal{S}_{2} = {\langle{\mathrmbf{R}_{2},X_{2},\mathrmbfit{S}_{2}}\rangle} \rightarrow 
{\langle{\mathrmbf{R}_{1},X_{1},\mathrmbfit{S}_{1}}\rangle} = \mathcal{S}_{1}$
consists of 
a relation functor $\mathrmbfit{F} : \mathrmbf{R}_{2} \rightarrow \mathrmbf{R}_{1}$,
a function on entity types $f : X_{2} \rightarrow X_{1}$ and
a signature natural transformation
$\theta : \mathrmbfit{S}_{2} \circ {\scriptstyle\sum}_{f} \Rightarrow \mathrmbfit{F} \circ \mathrmbfit{S}_{1}$.
Any (strict) relational database schema morphism
${\langle{\mathrmbfit{F},f}\rangle} : \mathcal{S}_{2} \rightarrow \mathcal{S}_{1}$
determines a type language morphism 
$\mathrmbfit{lang}(\mathrmbfit{F},f) 
: \mathrmbfit{lang}(\mathrmbfit{S}_{2}) \rightarrow \mathrmbfit{lang}(\mathrmbfit{S}_{1})$
in the Information Flow Framework~\cite{iff},
since we have the commutative diagram
$\mathrmbfit{F} \circ \partial_{1} = \partial_{2} \circ \mathrmbfit{sign}(\coprod\mathrmbfit{F},f)$.}
an entity infomorphism 
${\langle{f,g}\rangle} : 
\mathcal{E}_{2} = {\langle{X_{2},Y_{2},\models_{\mathcal{E}_{2}}}\rangle} \rightleftarrows 
{\langle{X_{1},Y_{1},\models_{\mathcal{E}_{1}}}\rangle} = \mathcal{E}_{1}$
with a common (entity) type function component $f : X_{2} \rightarrow X_{1}$
and a universe (entity instance) function $g : X_{1} \rightarrow X_{2}$, and
a key natural transformation
$\kappa : \mathrmbfit{F}^{\mathrm{op}} \circ \mathrmbfit{K}_{1} \Rightarrow \mathrmbfit{K}_{2}$,
which satisfy the condition
$\kappa \bullet \tau_{2}
= \mathrmbfit{F}^{\mathrm{op}}\tau_{1} \bullet \theta^{\mathrm{op}}_{{\langle{f,g}\rangle}}$,
where
$\theta^{\mathrm{op}}_{{\langle{f,g}\rangle}} \doteq
\theta^{\mathrm{op}}\mathrmbfit{tup}_{\mathcal{E}_{1}} \bullet \mathrmbfit{S}_{2}^{\mathrm{op}}\tau_{{\langle{f,g}\rangle}}$.
\footnote{$\tau_{{\langle{f,g}\rangle}}$ 
is the tuple natural transformation in the morphism of presheaves
$\mathrmbfit{tup}(f,g) = {\langle{{\scriptstyle\sum}_{f},f^{\ast},\tau_{{\langle{f,g}}\rangle}}\rangle} : 
{\langle{(\mathrmbf{Set}{\downarrow}X_{2}),\mathrmbfit{tup}_{\mathcal{E}_{2}}}\rangle} \rightleftarrows
{\langle{(\mathrmbf{Set}{\downarrow}X_{1}),\mathrmbfit{tup}_{\mathcal{E}_{1}}}\rangle}$
coming from the tuple functor
$\mathrmbfit{tup} : \mathrmbf{Cls} \rightarrow (\mathrmbf{Adj}{\Downarrow}\mathrmbf{Set})$.}
It is strict or trim when 
the underlying relational database schema morphism is strict or trim
($\theta = 1$).
Figure~\ref{relational:database:morphism} illustrates in detail a relational database morphism.
Here are some examples of relational database morphisms.
\begin{description}
\item[Table morphism.]
A relational database morphism
${\langle{\mathrmbfit{F},\theta,f,g,\kappa}\rangle}$
with one-object source and target categories of relations
and identity relation functor
$\mathrmbfit{F} = \mathrmbfit{id}_{\mathrmbf{1}} : \mathrmbf{1} \rightarrow \mathrmbf{1}$
is identical to a single morphism of tables
${\langle{k,h,f,g}\rangle} : 
{\langle{K_{1},t_{1},I_{1},s_{1},\mathcal{E}_{1}}\rangle} \rightarrow 
{\langle{K_{2},t_{2},I_{2},s_{2},\mathcal{E}_{2}}\rangle}$,
except that the direction has switched.
\item[Infomorphism.] 
An entity infomorphism 
${\langle{f,g}\rangle} : 
\mathcal{E}_{2} = {\langle{X_{2},Y_{2},\models_{\mathcal{E}_{2}}}\rangle} \rightleftarrows 
{\langle{X_{1},Y_{1},\models_{\mathcal{E}_{1}}}\rangle} = \mathcal{E}_{1}$
is a relational database morphism
between the classifications regarded as relational databases,
where the following hold.
~\footnote{Since any pair of adjoint monotonic functions 
${\langle{f,g}\rangle} : {\langle{P_{2},\leq_{2}}\rangle} \rightleftarrows {\langle{P_{1},\leq_{1}}\rangle}$
is an infomorphism 
${\langle{f,g}\rangle} : {\langle{P_{2},P_{2},\leq_{2}}\rangle} \rightleftarrows {\langle{P_{1},P_{1},\leq_{1}}\rangle}$,
such a pair is a relational database morphism.}
%
%
The type function is monotonic
$f : {\langle{X_{2},\geq_{\mathcal{E}_{2}}}\rangle} \rightarrow {\langle{X_{1},\geq_{\mathcal{E}_{1}}}\rangle}$
mapping an ordering 
$x_{2}' \geq_{\mathcal{E}_{2}} x_{2}$ 
with $\mathrmbfit{ext}_{\mathcal{E}_{2}}(x_{2}') \supseteq \mathrmbfit{ext}_{\mathcal{E}_{2}}(x_{2})$
to the ordering $f(x_{2}') \geq_{\mathcal{E}_{1}} f(x_{2})$,
with $\mathrmbfit{ext}_{\mathcal{E}_{1}}(f(x_{2}')) \supseteq \mathrmbfit{ext}_{\mathcal{E}_{1}}(f(x_{2}))$,
since
$y_{1} \models_{\mathcal{E}_{1}} f(x_{2})$ iff 
$g(y_{1}) \models_{\mathcal{E}_{2}} x_{2}$ implies 
$g(y_{1}) \models_{\mathcal{E}_{2}} x_{2}'$ iff
$y_{1} \models_{\mathcal{E}_{1}} f(x_{2}')$.
%
For each type $x_{2} \in X_{2}$,
the type function $X_{2} \xrightarrow{f} X_{1}$
restricts to an arity function
${\langle{{\uparrow}_{\mathcal{E}_{1}}(x_{2}),\mathrmbfit{inc}_{x_{2}}{\cdot}f}\rangle}
\xrightarrow{f} {\langle{{\uparrow}_{\mathcal{E}_{1}}(f(x_{2})),\mathrmbfit{inc}_{f(x_{2})}}\rangle}$.
This is the $x_{2}^{\mathrm{th}}$-component function
of a signature natural transformation
$\theta : {\uparrow}_{\mathcal{E}_{2}} \circ {\scriptstyle\sum}_{f} \Rightarrow f \circ {\uparrow}_{\mathcal{E}_{1}}$.
%
%
For each type $x_{2} \in X_{2}$,
the instance function $Y_{1} \xrightarrow{g} Y_{2}$
restricts to a function
$\mathrmbfit{ext}_{\mathcal{E}_{1}}(f(x_{2})) \xrightarrow{g} \mathrmbfit{ext}_{\mathcal{E}_{2}}(x_{2})$,
since
 an instance $y_{1} \in Y_{1}$ satisfying $y_{1} \models_{\mathcal{E}_{1}} f(x_{2})$
determines
the instance $g(y_{1}) \in Y_{2}$ satisfying $g(y_{1}) \models_{\mathcal{E}_{2}} x_{2}$.
This is the
$x_{2}^{\mathrm{th}}$-component function
of a key natural transformation
$\mathrmbfit{g} : 
{\langle{X_{2},\leq_{\mathcal{E}_{2}}}\rangle} 
\xrightarrow{f^{\mathrm{op}}} {\langle{X_{1},\leq_{\mathcal{E}_{1}}}\rangle}
\xrightarrow{\mathrmbfit{ext}_{\mathcal{E}_{1}}} \mathrmbf{Set}
\Rightarrow 
{\langle{X,\leq_{\mathcal{E}_{2}}}\rangle} \xrightarrow{\mathrmbfit{ext}_{\mathcal{E}_{2}}} \mathrmbf{Set}$.
%
%
Finally,
for any entity type $x_{2} \in X_{2}$,  
the two functions
{\footnotesize$
\mathrmbfit{ext}_{\mathcal{E}_{1}}(f(x_{2})) 
\xrightarrow{g} 
\mathrmbfit{ext}_{\mathcal{E}_{2}}(x_{2}) 
\xrightarrow{\Delta} 
\mathrmbfit{tup}_{\mathcal{E}_{2}}({\uparrow}x_{2},\mathrmbfit{inc}_{x_{2}}) 
$}
and
{\footnotesize$
\mathrmbfit{ext}_{\mathcal{E}_{1}}(f(x_{2})) 
\xrightarrow{\Delta}
\mathrmbfit{tup}_{\mathcal{E}_{1}}({\uparrow}f(x_{2}),\mathrmbfit{inc}_{f(x_{2})}) 
\xrightarrow{f{\cdot}{(\mbox{-})}}
\mathrmbfit{tup}_{\mathcal{E}_{1}}({\uparrow}x_{2},\mathrmbfit{inc}_{x_{2}}{\cdot}f) 
\xrightarrow{{(\mbox{-})} \cdot g}
\mathrmbfit{tup}_{\mathcal{E}_{2}}({\uparrow}x_{2},\mathrmbfit{inc}_{x_{2}}) 
$}
are equal.
%
%
\item[IFF Structure Morphism.]
Using the key natural transformation
$\kappa : \mathrmbfit{F}^{\mathrm{op}} \circ \mathrmbfit{K}_{1} \Rightarrow \mathrmbfit{K}_{2}$,
we can define the relation infomorphism 
${\langle{F,K}\rangle} :
\mathcal{R}_{2} = {\langle{R_{2},K_{2},\models_{\mathcal{R}_{2}}}\rangle} \rightleftarrows 
{\langle{R_{1},K_{1},\models_{\mathcal{R}_{1}}}\rangle} = \mathcal{R}_{1}$ 
with 
type function $F = \mathrmbfit{obj}(\mathrmbf{F}) : R_{2} \rightarrow R_{1}$ and
instance function $K : K_{1} \rightarrow K_{2} : k_{1} \mapsto \kappa_{r_{2}}(k_{1})$
using the $r_{2}^{\mathrm{th}}$ component function
$\kappa_{r_{2}} : \mathrmbf{K}_{1}(r_{1}) \rightarrow \mathrmbf{K}_{2}(r_{2})$
for each key $k_{1} \in \mathrmbf{K}_{1}(r_{1})$
of an image relation type $r_{1} = \mathrmbf{F}(r_{2})$ 
($K$ is defined by arbitrary choice, otherwise).
Any strict relational database morphism
${\langle{\mathrmbfit{F},f,g,\kappa}\rangle} :
\mathcal{D}_{2} = {\langle{\mathcal{S}_{2},\mathcal{E}_{2},\mathrmbfit{K}_{2},\tau_{2}}\rangle} \rightarrow
{\langle{\mathcal{S}_{1},\mathcal{E}_{1},\mathrmbfit{K}_{1},\tau_{1}}\rangle} = \mathcal{D}_{1}$
determines an structure (model) morphism
$\mathrmbfit{struc}(\mathrmbfit{F},f,g,\kappa) 
: \mathrmbfit{struc}(\mathrmbfit{S}_{2}) \rightarrow \mathrmbfit{struc}(\mathrmbfit{S}_{1})$
in the Information Flow Framework~\cite{iff}
with type language morphism 
$\mathrmbfit{lang}(\mathrmbfit{F},f) 
: \mathrmbfit{lang}(\mathrmbfit{S}_{2}) \rightarrow \mathrmbfit{lang}(\mathrmbfit{S}_{1})$,
entity infomorphism 
${\langle{f,g}\rangle} : 
\mathcal{E}_{2} = {\langle{X_{2},Y_{2},\models_{\mathcal{E}_{2}}}\rangle} \rightleftarrows 
{\langle{X_{1},Y_{1},\models_{\mathcal{E}_{1}}}\rangle} = \mathcal{E}_{1}$,
semidesignation morphism
${\langle{\coprod\mathrmbfit{F},f,g}\rangle} : 
{\langle{I_{2},s_{2},\mathcal{E}_{2}}\rangle} \rightarrow {\langle{I_{1},s_{1},\mathcal{E}_{1}}\rangle}$
and
relation infomorphism 
${\langle{F,K}\rangle} :
\mathcal{R}_{2} \rightleftarrows \mathcal{R}_{1}$. 
\end{description}
Composition of morphisms is defined component-wise.
Let $\mathrmbf{Db}$ denote the category of relational databases 
with the two projections ($\mathrmbfit{dbs}$ is called the schema functor)
$\mathrmbf{Dbs} \xleftarrow{\mathrmbfit{dbs}} \mathrmbf{Db}
\xrightarrow{\mathrmbfit{cls}} \mathrmbf{Cls}$
and the key functor
$\mathrmbf{Db} \xrightarrow{\mathrmbfit{key}} (\mathrmbf{Cat}{\Downarrow}\mathrmbf{Set})$
mapping $\mathcal{D}$ to ${\langle{\mathrmbf{R},\mathrmbfit{K}}\rangle}$
and
${\langle{\mathrmbfit{F},\theta,f,g,\kappa}\rangle} :
\mathcal{D}_{2} = {\langle{\mathcal{S}_{2},\mathcal{E}_{2},\mathrmbfit{K}_{2},\tau_{2}}\rangle} \rightarrow
{\langle{\mathcal{S}_{1},\mathcal{E}_{1},\mathrmbfit{K}_{1},\tau_{1}}\rangle} = \mathcal{D}_{1}$
to ${\langle{\mathrmbfit{F},\kappa}\rangle} : 
{\langle{\mathrmbf{R}_{2},\mathrmbfit{K}_{2}}\rangle} \rightarrow
{\langle{\mathrmbf{R}_{1},\mathrmbfit{K}_{1}}\rangle}$.
\begin{center}
\begin{tabular}{c}
\\ \\ \\
\setlength{\unitlength}{0.55pt}
\begin{picture}(320,80)(0,-10)
\put(0,80){\makebox(0,0){\footnotesize{$(\mathrmbf{Cat}{\Downarrow}\mathrmbf{Set})$}}}
\put(140,80){\makebox(0,0){\footnotesize{$\mathrmbf{Db}$}}}
\put(280,80){\makebox(0,0){\footnotesize{$\mathrmbf{Cls}$}}}
\put(0,0){\makebox(0,0){\footnotesize{$\mathrmbf{Cat}$}}}
\put(140,0){\makebox(0,0){\footnotesize{$\mathrmbf{Dbs}$}}}
\put(280,0){\makebox(0,0){\footnotesize{$\mathrmbf{Set}$}}}
\put(210,40){\makebox(0,0){\footnotesize{$\mathrmbf{SDsgn}$}}}
\put(80,92){\makebox(0,0){\scriptsize{$\mathrmbfit{key}$}}}
\put(210,92){\makebox(0,0){\scriptsize{$\mathrmbfit{cls}$}}}
\put(70,-12){\makebox(0,0){\scriptsize{$\mathrmbfit{rel}$}}}
\put(210,-12){\makebox(0,0){\scriptsize{$\mathrmbfit{typ}$}}}
\put(-8,40){\makebox(0,0)[r]{\scriptsize{$\mathrmbfit{top}$}}}
\put(132,40){\makebox(0,0)[r]{\scriptsize{$\mathrmbfit{sch}$}}}
\put(288,40){\makebox(0,0)[l]{\scriptsize{$\mathrmbfit{typ}$}}}
\put(140,115){\makebox(0,0){\tiny{${\langle{\mathcal{S}_{2},\mathcal{E}_{2},\mathrmbfit{K}_{2},\tau_{2}}\rangle} 
\xrightarrow{{\langle{\mathrmbfit{F},\theta,f,g,\kappa}\rangle}}
{\langle{\mathcal{S}_{1},\mathcal{E}_{1},\mathrmbfit{K}_{1},\tau_{1}}\rangle}$}}}
\put(-60,80){\makebox(0,0)[r]{\tiny{${\langle{\mathrmbf{R}_{2},\mathrmbfit{K}_{2}}\rangle} \xrightarrow{{\langle{\mathrmbfit{F},\kappa}\rangle}}{\langle{\mathrmbf{R}_{1},\mathrmbfit{K}_{1}}\rangle}$}}}
\put(310,80){\makebox(0,0)[l]{\tiny{${\langle{X_{2},Y_{2},\models_{\mathcal{E}_{2}}}\rangle} \xrightarrow{{\langle{f,g}\rangle}}{\langle{X_{1},Y_{1},\models_{\mathcal{E}_{1}}}\rangle}$}}}
\put(140,-25){\makebox(0,0){\tiny{$
{\langle{\mathrmbf{R}_{2},X_{2},\mathrmbfit{S}_{2}}\rangle} \xrightarrow{{\langle{\mathrmbfit{F},\theta,f}\rangle}}{\langle{\mathrmbf{R}_{1},X_{1},\mathrmbfit{S}_{1}}\rangle}$}}}
\put(-80,0){\makebox(0,0)[r]{\tiny{$\mathrmbf{R}_{2}\xrightarrow{\mathrmbfit{F}}\mathrmbf{R}_{1}$}}}
\put(360,0){\makebox(0,0)[l]{\tiny{$X_{2}\xrightarrow{f}X_{1}$}}}
\put(115,80){\vector(-1,0){60}}
\put(165,80){\vector(1,0){90}}
\put(115,0){\vector(-1,0){90}}
\put(165,0){\vector(1,0){90}}
\put(0,65){\vector(0,-1){50}}
\put(140,65){\vector(0,-1){50}}
\put(280,65){\vector(0,-1){50}}
\put(155,70){\vector(2,-1){36}}
\put(190,30){\vector(-2,-1){36}}
\put(230,50){\vector(2,1){36}}
\qbezier(230,20)(235,20)(240,20)
\qbezier(240,20)(240,25)(240,30)
\end{picture}
\\ \\
\end{tabular}
\end{center}
%


%
\begin{figure}
\begin{center}
\begin{tabular}{@{\hspace{-50pt}}c}
\\ \\
\begin{tabular}{@{\hspace{-70pt}}c@{\hspace{140pt}}c}
\begin{tabular}{c}
\put(0,0){\setlength{\unitlength}{0.5pt}
\begin{picture}(120,100)(0,0)
\put(0,100){\makebox(0,0){\scriptsize{$\mathrmbf{R}_{2}^{\mathrm{op}}$}}}
\put(0,0){\makebox(0,0){\scriptsize{$\mathrmbf{Set}$}}}
\put(120,100){\makebox(0,0){\scriptsize{$\mathrmbf{R}_{1}^{\mathrm{op}}$}}}
\put(120,0){\makebox(0,0){\scriptsize{$\mathrmbf{Set}$}}}
\put(-30,50){\makebox(0,0)[r]{\scriptsize{$\mathrmbfit{K}_{2}$}}}
\put(30,50){\makebox(0,0)[l]{\scriptsize{$\mathrmbfit{tup}_{2}$}}}
\put(90,50){\makebox(0,0)[r]{\scriptsize{$\mathrmbfit{K}_{1}$}}}
\put(150,50){\makebox(0,0)[l]{\scriptsize{$\mathrmbfit{tup}_{1}$}}}
\put(65,110){\makebox(0,0){\scriptsize{$\mathrmbfit{F}^{\mathrm{op}}$}}}
\put(60,-10){\makebox(0,0){\scriptsize{$\mathrmbfit{id}$}}}
\put(65,77){\makebox(0,0){\large{$\overset{\rule[-2pt]{0pt}{5pt}\;\;\;\;\theta^{\mathrm{op}}_{{\langle{f,g}\rangle}}}{\Longleftarrow}$}}}
\put(55,30){\makebox(0,0){\large{$\overset{\rule[-2pt]{0pt}{5pt}\kappa}{\Longleftarrow}$}}}
\put(0,55){\makebox(0,0){\large{$\overset{\rule[-2pt]{0pt}{5pt}\tau_{2}}{\Rightarrow}$}}}
\put(120,55){\makebox(0,0){\large{$\overset{\rule[-2pt]{0pt}{5pt}\tau_{1}}{\Rightarrow}$}}}
\put(20,100){\vector(1,0){80}}
\put(20,0){\vector(1,0){80}}
\qbezier(-15,85)(-40,50)(-15,15)\put(-15,15){\vector(2,-3){0}}
\qbezier(15,85)(40,50)(15,15)\put(15,15){\vector(-2,-3){0}}
\qbezier(105,85)(80,50)(105,15)\put(105,15){\vector(2,-3){0}}
\qbezier(135,85)(160,50)(135,15)\put(135,15){\vector(-2,-3){0}}
\end{picture}}
\end{tabular}
&
\begin{tabular}{c}
\put(0,0){\setlength{\unitlength}{0.5pt}
\begin{picture}(120,100)(0,-10)
\put(0,80){\makebox(0,0){\scriptsize{$\mathrmbfit{F}^{\mathrm{op}} \circ \mathrmbfit{K}_{1}$}}}
\put(0,0){\makebox(0,0){\scriptsize{$\mathrmbfit{K}_{2}$}}}
\put(125,94){\makebox(0,0){\scriptsize{$\overset{\textstyle\mathrmbfit{F}^{\mathrm{op}}\circ\mathrmbfit{tup}_{1}(\mathrmbf{R}_{1},\mathrmbfit{S}_{1})}{\overbrace{\mathrmbfit{F}^{\mathrm{op}} \circ \mathrmbfit{S}_{1}^{\mathrm{op}} \circ \mathrmbfit{tup}_{\mathcal{E}_{1}}}}$}}}
\put(125,-14){\makebox(0,0){\scriptsize{$\overset{\textstyle\underbrace{\mathrmbfit{S}_{2}^{\mathrm{op}}\circ\mathrmbfit{tup}_{\mathcal{E}_{2}}}}{\mathrmbfit{tup}_{2}(\mathrmbf{R}_{2},\mathrmbfit{S}_{2})}$}}}
\put(125,40){\makebox(0,0){\scriptsize{$\mathrmbfit{S}_{2}^{\mathrm{op}} \circ ({\scriptstyle\sum}_{f})^{\mathrm{op}} \circ \mathrmbfit{tup}_{\mathcal{E}_{1}}$}}}
\put(130,60){\makebox(0,0)[l]{\scriptsize{$\theta^{\mathrm{op}}\mathrmbfit{tup}_{\mathcal{E}_{1}}$}}}
\put(130,20){\makebox(0,0)[l]{\scriptsize{$\mathrmbfit{S}_{2}^{\mathrm{op}}{\tau}_{{\langle{f,g}\rangle}}$}}}
\put(47,95){\makebox(0,0){\scriptsize{$\mathrmbfit{F}^{\mathrm{op}}\tau_{1}$}}}
\put(47,-15){\makebox(0,0){\scriptsize{$\tau_{2}$}}}
\put(-10,40){\makebox(0,0)[r]{\scriptsize{$\kappa$}}}
\put(47,80){\makebox(0,0){\large{$\Rightarrow$}}}
\put(47,0){\makebox(0,0){\large{$\Longrightarrow$}}}
\put(120,60){\makebox(0,0){\large{$\Downarrow$}}}
\put(120,20){\makebox(0,0){\large{$\Downarrow$}}}
\put(0,40){\makebox(0,0){\large{$\Downarrow$}}}
\put(225,40){\makebox(0,0){\large{$\Downarrow$}}}
\put(235,54){\makebox(0,0)[l]{\scriptsize{$\overset{\textstyle\theta^{\mathrm{op}}_{{\langle{f,g}\rangle}}}
{\overbrace{\mathrmbfit{tup}_{{\langle{f,g}\rangle}}(\mathrmbfit{F},\theta)}}$}}}
\end{picture}}
\end{tabular}
\end{tabular}
\\ \\ \\
\setlength{\unitlength}{0.6pt}
\begin{picture}(120,240)(0,10)
%
\put(0,200){\makebox(0,0){\footnotesize{$\mathrmbf{R}_{2}^{\mathrm{op}}$}}}
\put(120,200){\makebox(0,0){\footnotesize{$\mathrmbf{R}_{1}^{\mathrm{op}}$}}}
\put(60,210){\makebox(0,0){\scriptsize{$\mathrmbfit{F}^{\mathrm{op}}$}}}
\put(20,200){\vector(1,0){80}}
\put(0,185){\vector(0,-1){90}}
\put(-16,188){\vector(-4,-3){48}}
\put(-64,128){\vector(4,-3){48}}
\put(-44,174){\makebox(0,0)[r]{\scriptsize{$\mathrmbfit{T}_{2}$}}}
\put(6,140){\makebox(0,0)[l]{\scriptsize{$\mathrmbfit{S}_{2}^{\mathrm{op}}$}}}
\put(-132,120){\makebox(0,0)[r]{\scriptsize{$\mathrmbfit{K}_{2}$}}}
\put(-44,108){\makebox(0,0)[r]{\scriptsize{$\mathrmbfit{sign}^{\mathrm{op}}_{\mathcal{E}_{2}}$}}}
\put(-67,60){\makebox(0,0)[r]{\shortstack{\scriptsize{$\mathrmbfit{key}_{\mathcal{E}_{2}}$}}}}
\put(120,185){\vector(0,-1){90}}
\put(136,188){\vector(4,-3){48}}
\put(184,128){\vector(-4,-3){48}}
\qbezier(-20,200)(-130,180)(-130,120)\qbezier(-130,120)(-130,30)(35,3)\put(35,3){\vector(4,-1){0}}
\qbezier(140,200)(250,180)(250,120)\qbezier(250,120)(250,30)(85,3)\put(85,3){\vector(-4,-1){0}}
\qbezier(-84,125)(-80,20)(40,10)\put(40,10){\vector(1,0){0}}
\qbezier(204,125)(200,20)(80,10)\put(80,10){\vector(-1,0){0}}
	\put(60,148.5){\makebox(0,0){\large{$\overset{\rule[-2pt]{0pt}{5pt}\;\;\;\;{\theta}^{\mathrm{op}}}{\Leftarrow}$}}}
\qbezier[50](0,140)(60,140)(120,140)
\put(35,115){\makebox(0,0){\large{$\overset{\rule[-2pt]{0pt}{5pt}\kappa}{\Longleftarrow}$}}}
\qbezier[150](-128,110)(60,110)(248,110)
\put(164,174){\makebox(0,0)[l]{\scriptsize{$\mathrmbfit{T}_{1}$}}}
\put(114,140){\makebox(0,0)[r]{\scriptsize{$\mathrmbfit{S}_{1}^{\mathrm{op}}$}}}
\put(252,120){\makebox(0,0)[l]{\scriptsize{$\mathrmbfit{K}_{1}$}}}
\put(164,108){\makebox(0,0)[l]{\scriptsize{$\mathrmbfit{sign}^{\mathrm{op}}_{\mathcal{E}_{1}}$}}}
\put(192,60){\makebox(0,0)[l]{\shortstack{\scriptsize{$\mathrmbfit{key}_{\mathcal{E}_{1}}$}}}}
\put(-37,60){\makebox(0,0){\shortstack{\scriptsize{$\tau_{\mathcal{E}}$}\\\large{$\Rightarrow$}}}}
\put(157,60){\makebox(0,0){\shortstack{\scriptsize{$\;\;\tau_{\mathcal{E}_{1}}$}\\\large{$\Leftarrow$}}}}
\put(-80,140){\makebox(0,0){\footnotesize{$(\mathrmbf{Set}{\downarrow}\mathrmbfit{tup}_{\mathcal{E}_{2}})$}}}
\put(200,140){\makebox(0,0){\footnotesize{$(\mathrmbf{Set}{\downarrow}\mathrmbfit{tup}_{\mathcal{E}_{1}})$}}}
\put(0,80){\makebox(0,0){\footnotesize{$(\mathrmbf{Set}{\downarrow}X_{2})^{\mathrm{op}}$}}}
\put(124,80){\makebox(0,0){\footnotesize{$(\mathrmbf{Set}{\downarrow}X_{1})^{\mathrm{op}}$}}}
\put(60,5){\makebox(0,0){\footnotesize{$\mathrmbf{Set}$}}}
\put(67,90){\makebox(0,0){\scriptsize{$({\scriptstyle\sum}_{f})^{\mathrm{op}}$}}}
\put(24,38){\makebox(0,0)[r]{\scriptsize{$\mathrmbfit{tup}_{\mathcal{E}_{2}}$}}}
\put(97,38){\makebox(0,0)[l]{\scriptsize{$\mathrmbfit{tup}_{\mathcal{E}_{1}}$}}}
	\put(60,55){\makebox(0,0){\shortstack{\scriptsize{$\;\;\;{\tau}_{{\langle{f,g}\rangle}}$}\\\large{$\Leftarrow$}}}}
\qbezier[25](25,48)(60,48)(95,48)
\put(35,80){\vector(1,0){50}}
\put(9,68){\vector(3,-4){40}}
\put(111,68){\vector(-3,-4){40}}
\put(0,230){\makebox(0,0){\scriptsize{$r_{2}' \xrightarrow{p_{2}} r_{2}$}}}
\put(120,244){\makebox(0,0){\scriptsize{$\overset{\textstyle{r'_{1}}}{\overbrace{\mathrmbfit{F}(r'_{2})}}
\xrightarrow{\overset{\textstyle{p_{1}}}{\overbrace{\mathrmbfit{F}(p_{2})}}} 
\overset{\textstyle{r_{1}}}{\overbrace{\mathrmbfit{F}(r_{2})}}$}}}
\put(-290,120){\setlength{\unitlength}{0.45pt}
\begin{picture}(0,0)(0,0)
\put(-10,80){\makebox(0,0){\scriptsize{$\mathrmbfit{K}_{2}(r_{2})$}}}
\put(120,80){\makebox(0,0){\scriptsize{$\mathrmbfit{tup}_{\mathcal{E}_{2}}(I_{2},s_{2})$}}}
\put(-10,0){\makebox(0,0){\scriptsize{$\mathrmbfit{K}(r_{2}')$}}}
\put(120,0){\makebox(0,0){\scriptsize{$\mathrmbfit{tup}_{\mathcal{E}_{2}}(I'_{2},s'_{2})$}}}
\put(45,95){\makebox(0,0){\tiny{$\tau_{2}(r_{2})$}}}
\put(45,-15){\makebox(0,0){\tiny{$\tau_{2}(r_{2}')$}}}
\put(-8,40){\makebox(0,0)[r]{\tiny{$\mathrmbfit{K}_{2}(p_{2})$}}}
\put(130,40){\makebox(0,0)[l]{\tiny{$\mathrmbfit{tup}_{\mathcal{E}_{2}}(h_{2})$}}}
\put(0,65){\vector(0,-1){50}}
\put(120,65){\vector(0,-1){50}}
\put(25,80){\vector(1,0){35}}
\put(25,0){\vector(1,0){35}}
\end{picture}}
\put(310,120){\setlength{\unitlength}{0.45pt}
\begin{picture}(0,0)(0,0)
\put(-10,80){\makebox(0,0){\scriptsize{$\mathrmbfit{K}_{1}(r_{1})$}}}
\put(120,80){\makebox(0,0){\scriptsize{$\mathrmbfit{tup}_{\mathcal{E}_{1}}(I_{1},s_{1})$}}}
\put(-10,0){\makebox(0,0){\scriptsize{$\mathrmbfit{K}_{1}(r'_{1})$}}}
\put(120,0){\makebox(0,0){\scriptsize{$\mathrmbfit{tup}_{\mathcal{E}_{1}}(I'_{1},s'_{1})$}}}
\put(45,95){\makebox(0,0){\tiny{$\tau_{1}(r_{1})$}}}
\put(45,-15){\makebox(0,0){\tiny{$\tau_{1}(r'_{1})$}}}
\put(-8,40){\makebox(0,0)[r]{\tiny{$\mathrmbfit{K}_{1}(p_{1})$}}}
\put(130,40){\makebox(0,0)[l]{\tiny{$\mathrmbfit{tup}_{\mathcal{E}_{1}}(h_{1})$}}}
\put(0,65){\vector(0,-1){50}}
\put(120,65){\vector(0,-1){50}}
\put(25,80){\vector(1,0){40}}
\put(25,0){\vector(1,0){40}}
\end{picture}}
\put(-180,70){\makebox(0,0)[r]{\scriptsize{$\overset{\textstyle\mathrmbfit{S}_{2}(r_{2}')}{\overbrace{{\langle{I'_{2},s'_{2}}\rangle}}}\xrightarrow{\overset{\textstyle\mathrmbfit{S}_{2}(p_{2})}{\overbrace{h_{2}}}}\overset{\textstyle\mathrmbfit{S}_{2}(r_{2})}{\overbrace{{\langle{I_{2},s_{2}}\rangle}}}$}}}
\put(295,70){\makebox(0,0)[l]{\scriptsize{$\overset{\textstyle\mathrmbfit{S}_{1}(r'_{1})}{\overbrace{{\langle{I'_{1},s'_{1}}\rangle}}}\xrightarrow{\overset{\textstyle\mathrmbfit{S}_{1}(p_{1})}{\overbrace{h_{1}}}}\overset{\textstyle\mathrmbfit{S}_{1}(r_{1})}{\overbrace{{\langle{I_{1},s_{1}}\rangle}}}$}}}
\put(-260,15){\setlength{\unitlength}{0.45pt}\begin{picture}(0,0)(0,0)
\put(60,20){\makebox(0,0){\scriptsize{$\mathrmbfit{S}_{2} \circ {\scriptstyle\sum}_{f} \stackrel{\theta}{\Rightarrow} \mathrmbfit{F} \circ \mathrmbfit{S}_{1}$}}}
\put(60,-10){\makebox(0,0){\scriptsize{${\scriptstyle\sum}_{f}(I_{2},s_{2})={\langle{I_{2},s_{2}{\cdot}f}\rangle}\xrightarrow{{\theta}_{r_{2}}}{\langle{I_{1},s_{1}}\rangle}$}}}
\put(60,-40){\makebox(0,0){\scriptsize{$\mathrmbfit{tup}_{\mathcal{E}_{1}}(I_{1},s_{1})\xrightarrow{\theta_{r_{2}}{\cdot}{(\mbox{-})}}\mathrmbfit{tup}_{\mathcal{E}_{1}}(I_{2},s_{2}{\cdot}f)$}}}
\end{picture}}
\end{picture}
\\ \\ 
\\
{\scriptsize $\kappa \bullet \tau_{2} = 
\mathrmbfit{F}^{\mathrm{op}}\tau_{1} 
\bullet \theta^{\mathrm{op}}\mathrmbfit{tup}_{\mathcal{E}_{1}} 
\bullet \mathrmbfit{S}_{2}^{\mathrm{op}} \tau_{{\langle{f,g}\rangle}}$}
\\
{\scriptsize $\mathrmbfit{K}_{1}(r_{1}) 
\xrightarrow{\kappa_{r_{2}}} \mathrmbfit{K}_{2}(r_{2}) 
\xrightarrow{\tau_{2}(r_{2})} \mathrmbfit{tup}_{\mathcal{E}_{2}}(I_{2},s_{2}) =
\mathrmbfit{K}(r_{1}) 
\xrightarrow{\tau_{1}(r_{1})}
\mathrmbfit{tup}_{\mathcal{E}_{1}}(I_{1},s_{1}) 
\xrightarrow{
\overset{\mathrmbfit{tup}_{\mathcal{E}_{1}}(\theta_{r_{2}})}
{\theta_{r_{2}}{\cdot}{(\mbox{-})}}
}
\mathrmbfit{tup}_{\mathcal{E}_{1}}(I_{2},s_{2}{\cdot}f)
\xrightarrow{
\overset{\tau_{{\langle{f,g}\rangle}}(I_{2},s_{2})}
{{(\mbox{-})} \cdot g}
}
\mathrmbfit{tup}_{\mathcal{E}_{2}}(I_{2},s_{2})
$}
\\
\\ \\
{\scriptsize \begin{tabular}{p{260pt}}
This figure illustrates the condition on relational database morphisms.
It has been annotated to help guide the understanding.
The condition is symbolically stated in the two lines of text just above.
The top line states the condition in terms of natural transformations.
The bottom line states the condition in terms of set functions
on the $r_{2}^{\mathrm{th}}$ component for some source relation type $r_{2} \in \mathrmbf{R}_{2}$.
The large diagram in the center illustrates the condition.
The two upper diagrams give alternate views of this.
The top right diagram is in a form very much like a table morphism.
This is appropriate,
since a relational database morphism between single table databases is just a table morphism.
Finally,
we have illustrated the effect of the morphism on the source/target tables,
starting with a source relational constraint (morphism) $r_{2}' \xrightarrow{p} r_{2}$.
\end{tabular}}
\end{tabular}
\end{center}
\caption{Relational Database Morphism}
\label{relational:database:morphism}
\end{figure}
%


%
\begin{proposition}
There is a diagram functor
$\mathrmbf{Db} \xrightarrow{\mathrmbfit{dgm}} \left(\mathrmbf{Cat}{\,\Downarrow\,}\mathrmbf{Tbl}\right)$ 
from databases
to (the lax comma category of) diagrams of tables.
\end{proposition}
%
%
Recall that the limit operation is a functor
${\left(\mathrmbf{Cat}{\,\Downarrow\,}\mathrmbf{Tbl}\right)}^{\mathrm{op}} \xrightarrow{\mathrmbfit{lim}} \mathrmbf{Tbl}$.
\begin{definition}
The join functor is defined to be the composition 
{\footnotesize\[
\mathrmbfit{join} = \mathrmbfit{dgm}^{\mathrm{op}} \circ \mathrmbfit{lim} :
{\mathrmbf{Db}}^{\mathrm{op}} \rightarrow \mathrmbf{Tbl}.
\]\normalsize}
\end{definition}
%
%
\begin{corollary}
The schema of the join of a database is the
reference (colimit) of the underlying database schema. 
\begin{center}
\begin{tabular}{c}
\\
\setlength{\unitlength}{0.55pt}
\begin{picture}(120,100)(80,10)
\put(119,120){\makebox(0,0){\footnotesize{${\left(\mathrmbf{Cat}{\,\Downarrow\,}\mathrmbf{Tbl}\right)}$}}}
\put(50,110){\makebox(0,0)[r]{\scriptsize{$\mathrmbfit{dgm}$}}}
\put(190,110){\makebox(0,0)[l]{\scriptsize{$\mathrmbfit{lim}^{\mathrm{op}}$}}}
\put(20,90){\vector(3,1){54}}
\put(166,110){\vector(3,-1){54}}
%
\put(125,40){\makebox(0,0){\footnotesize{${\left(\mathrmbf{Cat}{\,\Uparrow\,}\mathrmbf{Sch}\right)}$}}}
\put(50,30){\makebox(0,0)[r]{\scriptsize{$\mathrmbfit{dgm}$}}}
\put(190,30){\makebox(0,0)[l]{\scriptsize{$\mathrmbfit{colim}$}}}
\put(129,80){\makebox(0,0)[l]{\scriptsize{$\widetilde{\mathrmbfit{sch}}$}}}
\put(20,10){\vector(3,1){54}}
\put(166,30){\vector(3,-1){54}}
\put(120,105){\vector(0,-1){50}}
\put(0,80){\makebox(0,0){\footnotesize{$\mathrmbf{Db}$}}}
\put(244,80){\makebox(0,0){\footnotesize{$\mathrmbf{Tbl}^{\mathrm{op}}$}}}
\put(0,0){\makebox(0,0){\footnotesize{$\mathrmbf{Dbs}$}}}
\put(240,0){\makebox(0,0){\footnotesize{$\mathrmbf{Sch}$}}}
\put(80,88){\makebox(0,0){\scriptsize{$\mathrmbfit{join}^{\mathrm{op}}$}}}
\put(120,-12){\makebox(0,0){\scriptsize{$\mathrmbfit{refer}$}}}
\put(-8,40){\makebox(0,0)[r]{\scriptsize{$\mathrmbfit{dbs}$}}}
\put(248,40){\makebox(0,0)[l]{\scriptsize{$\mathrmbfit{sch}$}}}
%
\put(25,80){\line(1,0){90}}
\put(155,80){\vector(1,0){60}}
\qbezier[20](115,80)(135,80)(155,80)
\put(25,0){\vector(1,0){190}}
\put(0,65){\vector(0,-1){50}}
\put(240,65){\vector(0,-1){50}}
%
%
\end{picture}
\\ \\
\end{tabular}
\end{center}
\end{corollary}
In any complete category,
the limits of arbirary diagrams
can be constructed by using only 
the terminal object and (binary) pullbacks.
Dually,
in any cocomplete category,
the colimits of arbirary diagrams
can be constructed by using only the initial object and (binary) pushouts.
As we have shown,
for any entity classification $\mathcal{E} = {\langle{X,Y,\models_{\mathcal{E}}}\rangle}$,
the category of $\mathcal{E}$-tables $\mathrmbf{Cat}(\mathcal{E})$ is complete.
Hence,
for any database schema $\mathcal{S}$
the join of arbitrary $\mathcal{S}$-databases 
can be constructed by using only 
the join of the empty database (the terminal $\mathcal{E}$-table) and 
the join of $\mathcal{E}$-databases with binary span $X$-schemas 
(two $\mathcal{E}$-tables connected through a third).
%

\section{Summary and Future Work}\label{summary:future:work}

We have define the semantics for the relational database logical environment,
which can be used to specify database system consequence.
This provides interpretations for various formalisms 
such as relational algebra and first order logic,
where terms and equations can be included
by replacing signature morphisms with (possibly quotiented) term-tuples.
The two most important acheivements of this paper are
the definition of a natural and general category of tables that is both complete and cocomplete,
and the definition of a morphism of databases with some very nice properties.
We have extended the notion of tables (Spivak~\cite{spivak:sd}),
first from an underlying entity type specification to an entity classification (models multi-inheritance),
second from the static case of an underlying entity classification
to the dynamic case of tables moving along an underlying entity infomorphism.
We have proven completeness and cocompleteness for this (larger) category of tables.
Completeness allows joins over arbitrary collections of tables that are possibly linked by projections.
This includes selection,
which is the join with respect to reference relations (tables). 
Cocompleteness allows a distributed union that is new.

However,
much work needs to be done.
We need to
investigate further properties of database morphisms, including continuity.
In a follow-up paper we will
develop various formalisms, such as relational algebra and first order logic, and define views and queries.
This will deepen the connection with the Information Flow Framework.
Functional dependencies and normal forms should be expressed in terms of the categorical structure.
For practical database maintenance,
modifications (insertion, deletion and update) need to be defined.
The unified form (plus its graphical representation) needs further development. 
And finally,
the theory of databases defined in this paper should be more closely compared and contrasted with other approachs, 
such as the simplicial database approach (Spivak~\cite{spivak:sd},\cite{spivak:fdm})
and the sketch approach (Johnson, Rosebrugh et al~\cite{johnson:rosebrugh:07}).

\end{document}